\documentclass[a4paper,UKenglish,cleveref, autoref]{lipics-v2021}

\title{Twin-width III: Max Independent Set, Min Dominating Set, and Coloring}
\titlerunning{Twin-width III: Maximum Independent Set, Min Dominating Set, and Coloring}

\author{\'{E}douard Bonnet}{Univ Lyon, CNRS, ENS de Lyon, Université Claude Bernard Lyon 1, LIP UMR5668, France \and \url{http://perso.ens-lyon.fr/edouard.bonnet/}}{edouard.bonnet@ens-lyon.fr}{https://orcid.org/0000-0002-1653-5822}{}
\author{Colin Geniet}{University of Warsaw}{colin.geniet@ens-paris-saclay.fr}{}{}
\author{Eun Jung Kim}{Universit\'{e} Paris-Dauphine, PSL University, CNRS UMR7243, LAMSADE, Paris, France}{eun-jung.kim@dauphine.fr}{https://orcid.org/0000-0002-6824-0516}{}
\author{St\'{e}phan Thomass\'{e}}{Univ Lyon, CNRS, ENS de Lyon, Universit\'{e} Claude Bernard Lyon 1, LIP UMR5668, France}{stephan.thomasse@ens-lyon.fr}{}{}
\author{R\'{e}mi Watrigant}{Univ Lyon, CNRS, ENS de Lyon, Universit\'{e} Claude Bernard Lyon 1, LIP UMR5668, France}{remi.watrigant@univ-lyon1.fr}{https://orcid.org/0000-0002-6243-5910}{}

\authorrunning{\'E. Bonnet, C. Geniet, E. J. Kim, S. Thomassé, R. Watrigant}

\Copyright{Édouard Bonnet, Colin Geniet, Eun Jung Kim, Stéphan Thomassé, Rémi Watrigant}

\ccsdesc[100]{Theory of computation → Graph algorithms analysis}
\ccsdesc[100]{Theory of computation → Fixed parameter tractability}

\keywords{Twin-width, Max Independent Set, Min Dominating Set, Coloring, Parameterized Algorithms, Approximation Algorithms, Exact Algorithms}

\category{}

\relatedversion{}

\supplement{}

\acknowledgements{}

\nolinenumbers 

\hideLIPIcs  

\EventEditors{John Q. Open and Joan R. Access}
\EventNoEds{2}
\EventLongTitle{42nd Conference on Very Important Topics (CVIT 2016)}
\EventShortTitle{CVIT 2016}
\EventAcronym{CVIT}
\EventYear{2016}
\EventDate{December 24--27, 2016}
\EventLocation{Little Whinging, United Kingdom}
\EventLogo{}
\SeriesVolume{42}
\ArticleNo{23}

\usepackage[utf8]{inputenc}  

\usepackage[T1]{fontenc}
\usepackage{lmodern}

\usepackage[colorinlistoftodos,bordercolor=orange,backgroundcolor=orange!20,linecolor=orange,textsize=normalsize]{todonotes}

\usepackage{amsmath}  
\usepackage{amssymb}     
\usepackage{bbm}

\usepackage{hyperref}
\usepackage{mathrsfs}
\usepackage{paralist}
\usepackage{fixmath}

\crefformat{equation}{\textup{#2(#1)#3}}
\crefrangeformat{equation}{\textup{#3(#1)#4--#5(#2)#6}}
\crefmultiformat{equation}{\textup{#2(#1)#3}}{ and \textup{#2(#1)#3}}
{, \textup{#2(#1)#3}}{, and \textup{#2(#1)#3}}
\crefrangemultiformat{equation}{\textup{#3(#1)#4--#5(#2)#6}}%
{ and \textup{#3(#1)#4--#5(#2)#6}}{, \textup{#3(#1)#4--#5(#2)#6}}{, and \textup{#3(#1)#4--#5(#2)#6}}

\Crefformat{equation}{#2Equation~\textup{(#1)}#3}
\Crefrangeformat{equation}{Equations~\textup{#3(#1)#4--#5(#2)#6}}
\Crefmultiformat{equation}{Equations~\textup{#2(#1)#3}}{ and \textup{#2(#1)#3}}
{, \textup{#2(#1)#3}}{, and \textup{#2(#1)#3}}
\Crefrangemultiformat{equation}{Equations~\textup{#3(#1)#4--#5(#2)#6}}%
{ and \textup{#3(#1)#4--#5(#2)#6}}{, \textup{#3(#1)#4--#5(#2)#6}}{, and \textup{#3(#1)#4--#5(#2)#6}}

\usepackage{xspace}
\usepackage{tikz}

\usepackage[ruled,vlined,linesnumbered]{algorithm2e}

\usetikzlibrary{fit}
\usetikzlibrary{arrows}
\usetikzlibrary{patterns}
\usetikzlibrary{calc}
\usetikzlibrary{shapes}
\usetikzlibrary{positioning}
\usetikzlibrary{math}
\usetikzlibrary{shapes.symbols}

\newcommand{\convexpath}[2]{
  [   
  create hullcoords/.code={
    \global\edef\namelist{#1}
    \foreach [count=\counter] \nodename in \namelist {
      \global\edef\numberofnodes{\counter}
      \coordinate (hullcoord\counter) at (\nodename);
    }
    \coordinate (hullcoord0) at (hullcoord\numberofnodes);
    \pgfmathtruncatemacro\lastnumber{\numberofnodes+1}
    \coordinate (hullcoord\lastnumber) at (hullcoord1);
  },
  create hullcoords
  ]
  ($(hullcoord1)!#2!-90:(hullcoord0)$)
  \foreach [
  evaluate=\currentnode as \previousnode using \currentnode-1,
  evaluate=\currentnode as \nextnode using \currentnode+1
  ] \currentnode in {1,...,\numberofnodes} {
    let \p1 = ($(hullcoord\currentnode) - (hullcoord\previousnode)$),
    \n1 = {atan2(\y1,\x1) + 90},
    \p2 = ($(hullcoord\nextnode) - (hullcoord\currentnode)$),
    \n2 = {atan2(\y2,\x2) + 90},
    \n{delta} = {Mod(\n2-\n1,360) - 360}
    in 
    {arc [start angle=\n1, delta angle=\n{delta}, radius=#2]}
    -- ($(hullcoord\nextnode)!#2!-90:(hullcoord\currentnode)$) 
  }
}

\usepackage[scr=boondox,scrscaled=1.05]{mathalfa}

\renewcommand{\geq}{\geqslant}
\renewcommand{\leq}{\leqslant}

\renewcommand{\le}{\leq}
\renewcommand{\ge}{\geq}

\newcommand{\card}[1]{|{#1}|}

\theoremstyle{definition}
\newenvironment{proofofclaim}{\noindent \textsc{Proof of the Claim:}}{\hfill$\Diamond$\medskip}

\newcommand{\clique}{\textsc{$k$-Clique}\xspace}
\newcommand{\kmis}{\textsc{$k$-Independent Set}\xspace}
\newcommand{\lmis}{\textsc{Max Independent Set}\xspace}
\newcommand{\wmis}{\textsc{Weighted Max Independent Set}\xspace}
\newcommand{\mis}{\textsc{MIS}\xspace}
\newcommand{\dmis}[1]{\textsc{Distance-#1 MIS}\xspace}
\newcommand{\kds}{\textsc{$k$-Dominating Set}\xspace}
\newcommand{\ds}{\textsc{Min Dominating Set}\xspace}
\newcommand{\rds}{\textsc{$(k,r)$-Dominating Set}\xspace}

\newcommand{\subiso}{\textsc{Subgraph Isomorphism}\xspace}
\newcommand{\indsub}{\textsc{Induced Subgraph Isomorphism}\xspace}
\newcommand{\scaset}{\textsc{$r$-Scattered Set}\xspace}

\newcommand{\sapsp}{\textsc{APSP}\xspace}
\newcommand{\apsp}{\textsc{All-Pairs Shortest Paths}\xspace}
\newcommand{\ssssp}{\textsc{SSSP}\xspace}
\newcommand{\sssp}{\textsc{Single-Source Shortest Paths}\xspace}

\newcommand{\algkmis}{\texttt{k-IndSet}\xspace}
\newcommand{\algindsub}{\texttt{IndSub}\xspace}
\newcommand{\algsubiso}{\texttt{SubIso}\xspace}
\newcommand{\algkds}{\texttt{k-DomSet}\xspace}

\newcommand{\algsssp}{\texttt{SSSP}\xspace}

\newcommand{\dec}{$\text{dec}$}
\newcommand{\best}{$\text{best}$}

\newcommand{\diam}{$\text{diam}$}

\newcommand{\tww}{tww}

\newcommand{\bip}{interval biclique partition\xspace}
\newcommand{\sbip}{IBP\xspace}
\newcommand{\bips}{interval biclique partitions\xspace}
\newcommand{\sbips}{IBPs\xspace}

\begin{document}

\maketitle

\begin{abstract}
  We recently introduced the notion of twin-width, a novel graph invariant, and showed that first-order model checking can be solved in time $f(d,k)n$ for $n$-vertex graphs given with a witness that the twin-width is at most~$d$, called $d$-contraction sequence or $d$-sequence, and formulas of size $k$~[Bonnet et al., FOCS '20].
  The inevitable price to pay for such a general result is that $f$ is a tower of exponentials of height roughly $k$. 
  In this paper, we show that algorithms based on twin-width need not be impractical.
  We present $2^{O(k)}n$-time algorithms for \textsc{$k$-Independent Set}, \textsc{$r$-Scattered Set}, \textsc{$k$-Clique}, and \textsc{$k$-Dominating Set} when an $O(1)$-sequence of the graph is given in input.
  We further show how to solve the weighted version of \textsc{$k$-Independent Set}, \textsc{Subgraph Isomorphism}, and \textsc{Induced Subgraph Isomorphism}, in the slightly worse running time $2^{O(k \log k)}n$.
  Up to logarithmic factors in the exponent, all these running times are optimal, unless the Exponential Time Hypothesis fails.
  Like our FO model checking algorithm, these new algorithms are based on a dynamic programming scheme following the sequence of contractions forward.
 
  We then show a second algorithmic use of the contraction sequence, by starting at its end and rewinding it.
  As an example of such a reverse scheme, we present a polynomial-time algorithm that properly colors the vertices of a graph with relatively few colors, thereby establishing that bounded twin-width classes are $\chi$-bounded.
  This significantly extends the $\chi$-boundedness of bounded rank-width classes, and does so with a very concise proof.
  It readily yields a constant approximation for \textsc{Max Independent Set} on $K_t$-free graphs of bounded twin-width, and a $2^{O(\text{OPT})}$-approximation for \textsc{Min Coloring} on bounded twin-width graphs.
  We further observe that a constant approximation for \textsc{Max Independent Set} on bounded twin-width graphs (but arbitrarily large clique number) would actually imply a PTAS.

  The third algorithmic use of twin-width builds on the second one.
  Playing the contraction sequence backward, we show that bounded twin-width graphs can be edge-partitioned into a linear number of bicliques, such that both sides of the bicliques are on consecutive vertices, in a fixed vertex ordering.
  This property is trivially shared with graphs of bounded average degree.
  Given that biclique edge-partition, we show how to solve the unweighted \textsc{Single-Source Shortest Paths} and hence \textsc{All-Pairs Shortest Paths} in sublinear time $O(n \log n)$ and time $O(n^2 \log n)$, respectively.
  In sharp contrast, even \textsc{Diameter} does not admit a truly subquadratic algorithm on bounded twin-width graphs, unless the Strong Exponential Time Hypothesis fails.

  The fourth algorithmic use of twin-width builds on the so-called \emph{versatile tree of contractions} [Bonnet et al., SODA '21], a branching and more robust witness of low twin-width.
  We present constant-approximation algorithms for \textsc{Min Dominating Set} and related problems, on bounded twin-width graphs, by showing that the integrality gap is constant.
  This is done by going down the versatile tree and stopping accordingly to a problem-dependent criterion.
  At the reached node, a greedy approach yields the desired approximation.   
\end{abstract}
\maketitle

\section{Introduction}\label{sec:intro}

As the title suggests, this is the third paper of a series~\cite{twin-width1,twin-width2} devoted to a new graph invariant called \emph{twin-width}.
All the results presented in this paper are self-contained as the relevant background is given in~\cref{sec:prelim}.
In the same section, the reader can find the definitions of \emph{contraction sequences} and \emph{twin-width}.
For now, we are content with some intuition on these notions.
This will be enough to sketch the ideas and techniques leading to our results, while sparing this introduction from too much formalism.

The twin-width of a graph is a non-negative integer measuring its distance to being a cograph.
Among the several characterizations of cographs, a possible definition goes as follows.
A graph is a \emph{cograph} if one can find therein two twins,\footnote{i.e., two vertices with the same neighborhood beside them} identify them, and iterate this process until there is only one vertex left.
Anticipating over the definitions of~\cref{sec:prelim}, this actually corresponds to a 0-sequence, witnessing that cographs have twin-width 0.
Conversely it is also true that graphs with twin-width 0 are cographs.
We generalize this identification process by allowing a controlled error on the contracted pairs of vertices.
An error graph or \emph{red graph} keeps the faulty adjacencies appearing between a contracted pair and the vertices that are neighbor of only one vertex of the pair.
A~$d$-sequence is an indentification or contraction sequence such that the maximum degree of the error graph never exceeds~$d$.
The existence of such a sequence entails that the initial graph has twin-width at most~$d$.

As it turns out, many graph classes have bounded twin-width: planar graphs and more generally proper minor-closed classes, bounded rank-width or clique-width graphs, proper hereditary subclasses of permutation graphs, unit interval graphs, and some particular class of cubic expanders, to name only a few.\footnote{A more exhaustive list is given in~\cref{thm:bd-tww}.}
Considering the wide variety of these classes, it might seem that our cograph generalization has gone too far to allow for a unified algorithmic treatment of bounded twin-width graphs.
The first paper of the series \cite{twin-width1} and the current one show that this is not the case.
Algorithms, whose running times are provably unattainable in general graphs, are actually possible in graphs of bounded twin-width.
We will now detail that point.

After defining any graph parameter $\kappa$, a natural question is whether some computationally hard problems can be solved more efficiently on graphs where $\kappa$ is bounded.
When this turns out to be the case for several problems, it may sometimes lead to a powerful meta-theorem.
A~standard way of capturing a large set of problems within the same framework is through the use of logic formulas over graphs, or more generally over relational structures.
In the language of parameterized algorithms, one may ask for the existence of a Fixed-Parameter Tractable (FPT) algorithm parameterized by $\kappa$ and the size of the graph formula $\varphi$ to be tested: More precisely, an algorithm deciding in time $f(|\varphi|, \kappa(G))n^{O(1)}$, or better $f(|\varphi|, \kappa(G))n$, whether an $n$-vertex graph $G$ satisfies $\varphi$, where $f$ is some computable function. 
Certainly the most famous result of that kind is the celebrated Courcelle's theorem, where the parameter $\kappa$ is tree-width, and the formula $\varphi$ ranges over Monadic Second Order logic (MSO$_2$) formulas~\cite{Courcelle90}.
On a slightly less general logic (namely MSO$_1$, where quantification over edge sets is disallowed), the result holds for the smaller parameter clique-width~\cite{Courcelle00}. 
It implies, for instance, that deciding whether a graph on $n$ vertices contains a subset of~$k$~pairwise non-adjacent vertices (i.e., solving \kmis) can be done in linear time on graphs of constant clique-width, while in general graphs it cannot be solved in polynomial time unless P$=$NP, or in time $f(k)n^{O(1)}$ unless FPT$=$W[1].
Such a result is unlikely for twin-width, as \kmis remains NP-hard in planar graphs which have constant twin-width. 
Nevertheless, when parameterized by the solution size $k$, an FPT algorithm is known in planar graphs, and more generally in any proper minor-closed graph class. 
Actually, on the latter class, every problem expressible by a first-order (FO) formula $\varphi$ can be solved in FPT time parameterized by~$|\varphi|$~\cite{Flum01}.
In the first paper of our series~\cite{twin-width1}, we extended this result and obtained the following meta-theorem for twin-width.

\begin{theorem}{\emph{\cite{twin-width1}}}\label{thm:firstorder}
Given an $n$-vertex graph $G$, a $d$-sequence of $G$, and a first-order formula~$\varphi$, one can decide $G \models \varphi$ in time $f(|\varphi|, d) n$ for some computable function~$f$.
\end{theorem}

The main drawback of this kind of algorithms is the obtained running time: The function $f$ is a tower of exponentials whose height depends on the size of the formula.
This is an unavoidable price to pay to solve at once all graph problems expressible in first-order logic.
Indeed, it is known that testing first-order formulas on trees requires a running time whose dependence in the size of the formula is a non-elementary function, unless P $=$ NP~\cite{Frick04}.
Furthermore the running time of our FO model checking algorithm does not get better on ``seemingly simpler'' formulas, such as for instance, with few quantifier alternations.

\paragraph*{Our results.}

We show that twin-width and its associated contraction sequence can also give rise to practical algorithms for some individual classic graph problems.
In particular, we consider the following NP-complete problems, given a graph $G$ and an integer $k$, decide if:
\begin{compactitem}
	\item \kmis: there are $k$ pairwise non-adjacent vertices.
	\item \clique: there are $k$ pairwise adjacent vertices.
	\item \scaset: there are $k$ vertices pairwise at distance at least $r$.
	\item \kds: there is a set $S$ of $k$ vertices such that for every vertex $v$ of $G$, either $v \in S$ or $v$ has a neighbor in $S$.
	\item \rds: there is a set $S$ of $k$ vertices such that every vertex of $G$ is at distance at most $r$ of some vertex in $S$.
\end{compactitem}

\medskip

These problems, parameterized by $k$, are W[1]-hard (the last two are even W[2]-complete), thus unlikely to admit an FPT algorithm, i.e., one with running time $f(k)n^{O(1)}$, on general graphs. 
We obtain single-exponential parameterized algorithms for all these problems when a contraction sequence witnessing ``twin-width at most $d$'' is given.
When considering the unparameterized optimization variant, we denote these five problems by \lmis (and \mis for short), \textsc{Max Clique}, \dmis{$(r-1)$}, \ds, and \textsc{Min $r$-Dominating Set}, respectively. 

\begin{theorem}\label{thm:single-exp}
  Given an $n$-vertex graph $G$ and a $d$-sequence $G=G_n, \ldots, G_1=K_1$, the above-mentioned five problems 
  can be solved in time $2^{O_d(k)}n$.
\end{theorem}


We then consider some W[1]-complete generalizations of \kmis or of \clique.
Namely:

\begin{compactitem}
	\item \wmis: given a graph $G$ with a weight function on vertices $w : V(G) \rightarrow \mathbb{R}$ and an integer $k$, decide whether there exists a set $S$ of size exactly $k$ of pairwise non-adjacent vertices such that $\sum_{v \in S} w(v)$ is maximum.
	\item \indsub: given a graph $H$ on $k$ vertices and a graph $G$, decide whether there exists a set $S \subseteq V(G)$ such that $G[S]$, the subgraph of $G$ induced by $S$, is isomorphic to $H$.
	\item \subiso: given a graph a graph $H$ on $k$ vertices and a graph $G$, decide whether there exists a set $S \subseteq V(G)$ such that $H$ is isomorphic to a subgraph of $G[S]$.
\end{compactitem}

\medskip

Unlike the other two problems, \subiso is \emph{not} a generalization of \kmis.
Though it does generalize \clique.
Once the formal definition of a contraction sequence is given, it will be clear that a $d$-sequence for $G$ readily yields a $d$-sequence for its complement, $\overline G$. 
Thus in the context of bounded twin-width graphs, an algorithm solving \subiso can be used to solve \kmis.
For these three problems, we now get slightly superexponential parameterized algorithms.

\begin{theorem}\label{thm:slightly-superexp}
  Given an $n$-vertex graph $G$ and a $d$-sequence $G=G_n, \ldots, G_1=K_1$, the above-mentioned three problems
    can be solved in time $2^{O_d(k \log k)}n$.
\end{theorem}

The algorithms behind \cref{thm:single-exp,thm:slightly-superexp} follow the same general plan.
Let us consider the $n$ successive red graphs $R_n, \ldots, R_1$ (error graphs) obtained after each vertex contraction.\footnote{A reader who would want precise definitions at this point is welcome to read first the couple of paragraphs of \cref{subsec:tww-def}.}
$R_n$ is the edgeless $n$-vertex graph (since there are initially no errors) and $R_1$ is the 1-vertex graph.
We maintain optimum partial solutions populating connected subgraphs of bounded size in each $R_i$.
Initially in $R_n$, the connected subgraphs are only made of single vertices (there are no edges).
So the optimum partial solutions are trivial to compute.
The partial solutions for $R_i$ are built from the partial solutions of $R_{i+1}$ in the following way.
Every partial solution \emph{not} involving the newly contracted vertex is simply kept.
Every partial solution involving the newly contracted vertex is computed by merging a bounded number of previous partial solutions on pairwise disconnected sets.
The key is that, by design, there is no error between the latter partial solutions.
Thus the presence or absence of edges can be decided regardless of the forgotten choices of precise vertices within the solution.
Eventually a (partial) solution is computed in $R_1$, which constitutes an actual solution in the entire initial graph $G$. 
In a nutshell, the algorithms may be summarized as dynamic programming over connected sets of the red graphs.

For \kmis there is not much more to it than the previous sketch.
For \textsc{(Induced) Subgraph Isomorphism} the algorithms become more technical.
Also conceptually, partial solutions are no longer necessarily feasible.
For \kds some new challenges appear.
The partial solutions and their actual specification are not straightforward to define, as it is for \kmis.

One may wonder if subexponential parameterized algorithms are possible for any of the eight problems considered so far.
We will observe that even \kmis cannot be solved in time $2^{o(k / \log k)}n^{O(1)}$ on graphs given with an $O(1)$-sequence, unless the Exponential Time Hypothesis fails.
With a similar argument, the same lower bound applies to \kds.
Thus, up to logarithmic factors in the exponent, the running times of \cref{thm:single-exp,thm:slightly-superexp} are optimal.
Actually we will see that even algorithms running in time $2^{o(n / \log n)}$ are unlikely.


\medskip

All the previous algorithms exploit the contraction sequence forward.
They follow the identification process from the initial graph $G$ to the 1-vertex graph.
What if we would start at the end, and maintain solutions as the vertices are iteratively split until the initial graph~$G$ is formed?
We exemplify the idea of using the contraction sequence backward with an essentially greedy coloring procedure that is not optimal but still uses relatively few colors.

Let us be more specific.
A proper $k$-coloring of a graph $G$ is a mapping $c : V(G) \rightarrow \{1, \dots, k\}$ such that $c(u) \neq c(v)$ whenever $uv \in E(G)$.
The chromatic number, denoted by $\chi(G)$, is the smallest integer $k$ such that $G$ admits a proper $k$-coloring.
It can be seen that $\chi(G) \geqslant \omega(G)$, where $\omega(G)$ denotes the size of a largest clique in $G$, whereas many constructions of triangle-free (that is, with $\omega(G) \leqslant 2$) graphs $G$ with arbitrarily large $\chi(G)$ are known.
A class of graph $\mathcal C$ is said \emph{$\chi$-bounded} if there is a function $f$ such that for any graph $G \in \mathcal{C}$, we have $\chi(G) \leqslant f(\omega(G))$.
Our coloring algorithm $d+2$-color any triangle-free graph of twin-width at most~$d$, and more generally $(d+2)^{\omega(G)-1}$-color any graph $G$ given with a $d$-sequence.
In particular, it shows the following.

\begin{theorem}\label{thm:chibounded-informal}
  Every graph class with bounded twin-width is $\chi$-bounded.
\end{theorem}

Algorithmically this has some direct consequences for approximating the chromatic number, as well as, in the subcase of $K_t$-free graphs, the independence number.

\medskip

The same idea of considering the contraction sequence backward is then used to show that every graph given with an $O(1)$-sequence admits an edge partition by $O(n)$ bicliques, each side of which is on consecutive vertices, for a fixed vertex ordering.
We use this edge partition to tackle the edge-unweighted version of some classic polynomial-time solvable problems:
\begin{compactitem}
  \item \sssp: given a graph $G$ and a source $s$, find a shortest-path tree rooted at $s$, spanning the connected component of $s$.
  \item \apsp: given a graph $G$, find the distances in $G$ between every pair of vertices.
  \item \textsc{Diameter}: given a graph $G$, report the largest distance in $G$ between two vertices.
\end{compactitem}

\medskip

We show how breadth-first search (BFS) can be mimicked, when replacing ``traversing an edge'' by ``traversing a biclique all at once''.
A subtlety of the algorithm, beside the necessary data structures to get \sssp sublinear in the total number of edges, lies in the fact that bicliques, contrary to single edges, can be traversed twice (once in both directions) before being discarded. 

\begin{theorem}\label{thm:sssp-informal}
  If the input graph comes with an $O(1)$-sequence, \sssp can be solved in $O(n \log n)$ time, thus \apsp and \textsc{Diameter} can be solved in $O(n^2 \log n)$ time.
  In contrast, \textsc{Diameter} cannot be solved in $O(n^{2-\varepsilon})$ for any $\varepsilon > 0$, even in that scenario, unless the Strong Exponential Time Hypothesis fails.
\end{theorem}

Our algorithm inherently relies on unweighted edges.
Nonetheless vertex-weights can be supported with the same running time.

\ds is known to be as approximable as the \textsc{Set Cover} problem. 
Thus, by classic papers by Johnson~\cite{Johnson74} and by Lov\'asz~\cite{Lovasz75}, it admits a $\ln n$-approximation and the integrality gap (i.e., the ratio between the optimum of the original problem and the optimum of the LP relaxation) of its standard LP formulation is also $\ln n$.
In sharp contrast, unless P$=$NP, \ds cannot be approximated in polynomial-time within factor $(1-o(1))\ln n$ on $n$-vertex general graphs~\cite{Dinur14}.

We show that, on bounded twin-width classes, the integrality gap of \ds is constant.
This uses the \emph{versatile trees of contractions} developed in the second paper of the series~\cite{twin-width2}.
These are more robust witnesses of low twin-width which, instead of providing a single contraction in a given trigraph, gives linearly many disjoint ones.
Placing ourselves at a right node of the versatile tree, we show that a greedy strategy in the corresponding trigraph yields a constant approximation in the original graph.

\begin{theorem}\label{thm:approx-ds-etal}
  If the input graph comes with an $O(1)$-sequence, \ds, \dmis{2}, and more generally \textsc{Min $r$-Dominating Set}, \dmis{$2r$} for every positive $r$, admit $O(1)$-approximation algorithms. 
\end{theorem}
These results are particular cases of the fact that when the twin-width of a matrix $A$ is bounded, there is a linear gap between the packing number and the minimum hitting set of the hypergraph with incidence matrix $A$.
Bounded twin-width matrices might more generally provide linear programs with bounded duality gap. 
It is noteworthy that \lmis (which corresponds to \dmis{1}) is \emph{not} covered by the previous theorem.
We further give some evidence that \mis may have a very different approximability status that \ds on bounded twin-width graphs.

\paragraph*{Related work.}

It is intrinsically difficult to compare our work to the existing literature since bounded twin-width graphs cover a wide spectrum of graph classes (more precisely, see \cref{thm:bd-tww} in \cref{sec:prelim}) and is rather transversal to well-established graph classes (see in the same subsection which graphs \emph{are} and which graphs are \emph{not} of bounded twin-width).
We sample some data points showing that our algorithms fare well even when compared to the state-of-the-art on a particular class of bounded twin-width (think, a single item on the list of \cref{thm:bd-tww}).
In that respect, the most flattering comparison point for our algorithms is perhaps with \subiso and \indsub.
On the contrary, \kmis admits parameterized subexponential algorithms on several sparse classes~\cite{Demaine05}, an easy single-exponential algorithm on bounded-degeneracy graphs by bounded search tree, and polynomial-time algorithms on perfect graphs \cite{Grotschel81} and other classes~\cite{GrzesikKPP19}, with which we cannot hope to uniformly compete.

\indsub, and particularly \subiso, have a long history of parameterized algorithms on sparse classes.
Let us recall some steps of that history.
Eppstein showed how to solve \textsc{(Induced) Subgraph Isomorphism} in time $2^{O(k \log k)}n$ on planar graphs~\cite{Eppstein99}, and then on apex\footnote{An apex graph is one that can be made planar by removing a single vertex.}-minor free graphs~\cite{Eppstein00}.
The latter algorithm would later be shown to work on every proper minor-closed class of graphs.
In modern terms, Eppstein's algorithm is based on \emph{low treewidth colorings}, and more precisely on the fact that planar graphs, but more generally $H$-minor free graphs, can be $k+1$-colored so that the union of any $k$ color class has treewidth $O(k)$.
Introducing a new kind of dynamic programming, dubbed \emph{embedded}, Dorn \cite{Dorn10} improved the running time of solving \indsub on planar graphs to $2^{O(k)} n$.
More recently, Pilipczuk and Siebertz presented a polynomial-space $2^{O(k \log k)}n$-time algorithm for \indsub on $H$-minor free graphs~\cite{PilipczukS19}.
This mainly uses the treedepth counterpart of Eppstein's approach.

Given an $O(1)$-sequence, our algorithm for \textsc{(Induced) Subgraph Isomorphism} also runs in time $2^{O(k \log k)}n$ (while it may face dense graphs) for the far-reaching generalization of bounded twin-width graphs (again we refer the reader to~\cref{thm:bd-tww} for other examples of bounded twin-width classes).
We also show with an elementary one-and-a-half-page proof that bounded twin-width classes are $\chi$-bounded.
This can be put in perspective with the $\chi$-boundedness of graphs of bounded clique-width~\cite{Dvorak12}, which is not an easy result.

On general graphs, the current fastest algorithm for the vertex-weighted variant of \apsp (\sapsp) is due to Yuster and runs in time $O(n^{2.842})$~\cite{Yuster09}, while no truly subcubic (i.e., running in time $O(n^{3-\varepsilon})$) algorithm is known without the use of fast matrix multiplication.
Since \sssp (\ssssp) can easily be solved in time $O(n \log n)$ in sparse graphs, i.e., with $O(n)$ edges, the algorithm of \cref{thm:sssp-informal} is only relevant on bounded twin-width classes that are dense.
Among the dense classes of~\cref{thm:bd-tww}, one can find for example bounded clique-width graphs.
Recently Kratsch and Nelles showed how to solve vertex-weighted \sapsp on graphs given with a clique-width expression of width $\text{cw}$ in time $O(\text{cw}^2 n^2)$~\cite{Kratsch20}.

\paragraph*{Organization of the paper.}
In \cref{sec:prelim} we introduce the relevant graph-theoretic background, then formally define contraction sequences and twin-width, and finally summarize which classes are known to have bounded twin-width and explain how $d$-sequences are given to our forthcoming algorithms.
\cref{sec:kmis} contains a $2^{O(k)}n$-time algorithm for \kmis (and \scaset) and a $2^{O(k \log k)}n$-time algorithm for \textsc{(Induced) Subgraph Isomorphism}. 
In~\cref{sec:kds}, we present a $2^{O(k)}n$-time algorithm for \kds.
In~\cref{sec:chibounded}, we show that bounded twin-width classes are $\chi$-bounded and satisfy the strong Erd\H{o}s-Hajnal property.
In~\cref{sec:ibp-sp}, we prove that bounded twin-width graphs can be edge-partitioned into linearly many bicliques whose sides are both on consecutive vertices, for a fixed ordering of the vertex set.
We then use that property to derive algorithms solving \sssp and \apsp in time $O(n \log n)$ and $O(n^2 \log n)$, respectively.
We also observe that \textsc{Diameter} is unlikely to be solvable in truly subquadratic time, in graphs of bounded twin-width.
In~\cref{sec:approx-alg}, we give $O(1)$-approximation algorithms for \ds and related problems, provided a $d$-sequence.
We complement this result by some evidence that the approximability of \mis on bounded twin-width graphs may have a very different status.
Finally in~\cref{sec:conclusion}, we suggest some future work on approximation algorithms for bounded twin-width graphs and exact exponential algorithms for general graphs.

\section{Preliminaries}\label{sec:prelim}

We denote by $[i,j]$ the set of integers $\{i,i+1,\ldots, j-1, j\}$, and by $[i]$ the set of integers $[1,i]$.
If $\mathcal X$ is a set of sets, we denote by $\cup \mathcal X$ their union.
The notation $O_d(\cdot)$ gives an asymptotic behavior when $d$ is seen as a constant.
The notation $O^*(\cdot)$ suppresses polynomial factors.

Unless stated otherwise, all graphs are assumed undirected and simple, that is, they do not have parallel edges or self-loops.
We denote by $V(G)$ and $E(G)$, the set of vertices and edges, respectively, of a graph $G$. 
For $S \subseteq V(G)$, we denote the \emph{open neighborhood} (or simply \emph{neighborhood}) of $S$ by $N_G(S)$, i.e., the set of neighbors of $S$ deprived of $S$, and the \emph{closed neighborhood} of $S$ by $N_G[S]$, i.e., the set $N_G(S) \cup S$.
We simplify $N_G(\{v\})$ into $N_G(v)$, and $N_G[\{v\}]$ into $N_G[v]$.
We denote by $G[S]$ the subgraph of $G$ induced by $S$, and $G - S := G[V(G) \setminus S]$.
A \emph{connected subset} (or \emph{connected set}) $S \subseteq V(G)$ is one such that $G[S]$ is connected.
For two disjoint sets $A, B \subseteq V(G)$, $E(A,B)$ denotes the set of edges in $E(G)$ with one endpoint in $A$ and the other one in $B$.
We also denote by $G[A,B]$ the bipartite graph $(A \cup B,E(A,B))$.
Two distinct vertices $u, v$ such that $N(u) = N(v)$ are called \emph{false twins}, and \emph{true twins} if $N[u] = N[v]$.
Two vertices are \emph{twins} if they are false twins or true twins.
For two vertices $u, v \in V(G)$, the \emph{distance $d_G(u,v)$} is the number of edges in a shortest path from $u$ to $v$, and $\infty$ if $u$ and $v$ are in two distinct connected components of $G$.
Then the \emph{radius} of a graph $G$ is defined as $\min_{u \in V(G)} \max_{v \in V(G)} d_G(u,v)$ and the \emph{diameter $\diam(G)$} as $\max_{u \in V(G)} \max_{v \in V(G)} d_G(u,v)$. 
In all the notations with a graph subscript, we may omit it if the graph is clear from the context.

A graph is \emph{$H$-free} if it does not contain $H$ as an induced subgraph.
However we make an exception for $H = K_{t,t}$.
A~$K_{t,t}$-free graph is a graph with no biclique $K_{t,t}$ \emph{as a subgraph}.
An \emph{edge contraction}\footnote{Not to be confused with our (vertex) contractions, which can be on non-adjacent vertices.} of two adjacent vertices $u, v$ consists of merging $u$ and $v$ into a single vertex adjacent to $N(\{u,v\})$ (and deleting $u$ and $v$).
A graph $H$ is a \emph{minor} of a graph $G$ if $H$ can be obtained from $G$ by a sequence of vertex and edge deletions, and edge contractions.
A graph $G$ is said \emph{$H$-minor free} if $G$ does not contain $H$ as a minor.
A class\footnote{That is, a set of graphs closed under isomorphism} $\mathcal C$ of graphs \emph{has property $\Pi$} if every graph of $\mathcal C$ has property $\Pi$.
A class is \emph{hereditary} if it is closed under taking induced subgraphs.

\subsection{Trigraphs, contraction sequences, and twin-width of a graph}\label{subsec:tww-def}

A \emph{trigraph $G$} has vertex set $V(G)$, (black) edge set $E(G)$, and red edge set $R(G)$ (the error edges), with $E(G)$ and $R(G)$ being disjoint.
The \emph{set of neighbors $N_G(v)$} of a vertex $v$ in a trigraph $G$ consists of all the vertices adjacent to $v$ by a black or red edge.
A $d$-trigraph is a trigraph $G$ such that the \emph{red graph} $(V(G),R(G))$ has degree at most~$d$.
In that case, we also say that the trigraph has \emph{red degree} at most~$d$.
A (vertex) \emph{contraction} or \emph{identification} in a trigraph~$G$ consists of merging two (non-necessarily adjacent) vertices $u$ and $v$ into a single vertex $z$, and updating the edges of $G$ in the following way.
Every vertex of the symmetric difference $N_G(u) \triangle N_G(v)$ is linked to $z$ by a red edge.
Every vertex $x$ of the intersection $N_G(u) \cap N_G(v)$ is linked to $z$ by a black edge if both $ux \in E(G)$ and $vx \in E(G)$, and by a red edge otherwise.
The rest of the edges (not incident to $u$ or $v$) remain unchanged.
We insist that the vertices $u$ and $v$ (together with the edges incident to these vertices) are removed from the trigraph. 
See \cref{fig:contraction} for an illustration.
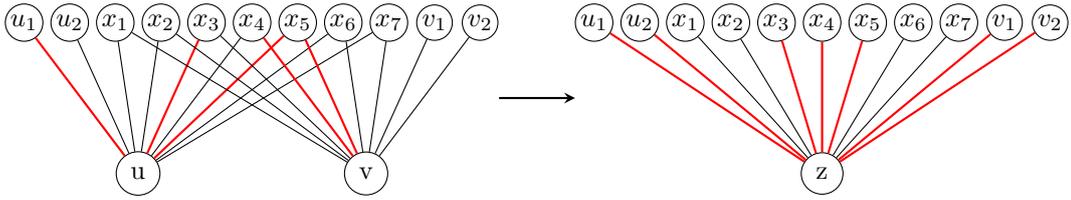
\begin{figure}
\begin{tikzpicture}
\def\v{2}
\def\t{6}
\def\s{0.6}

\draw[thick, -stealth] (3.25,\v /2) -- (4.25,\v/2) ;

\foreach \i/\j in {-5/u_1,-4/u_2,-3/x_1,-2/x_2,-1/x_3,0/x_4,1/x_5,2/x_6,3/x_7,4/v_1,5/v_2}{
  \node[draw,circle,inner sep=0.03cm] (n\i) at (\s * \i,\v) {$\j$} ; 
}

\node[draw,circle] (u) at (-1.5,0) {u} ;
\node[draw,circle] (v) at (1.5,0) {v} ;

\foreach \i in {-4,-3,-2,0,2,3}{
  \draw (u) -- (n\i) ;
}
\foreach \i in {-5,-1,1}{
  \draw[thick, red] (u) -- (n\i) ;
}
\foreach \i in {-3,...,-1,2,3,...,5}{
  \draw (v) -- (n\i) ;
}
\foreach \i in {0,1}{
  \draw[thick, red] (v) -- (n\i) ;
}

\begin{scope}[xshift=7.5cm]
\node[draw,circle] (uv) at (0,0) {z} ;
\foreach \i/\j in {-5/u_1,-4/u_2,-3/x_1,-2/x_2,-1/x_3,0/x_4,1/x_5,2/x_6,3/x_7,4/v_1,5/v_2}{
  \node[draw,circle,inner sep=0.03cm] (m\i) at (\s * \i,\v) {$\j$} ; 
}

\foreach \i in {-3,-2,2,3}{
  \draw (uv) -- (m\i) ;
}
\foreach \i in {-5,-4,-1,0,1,4,5}{
  \draw[thick, red] (uv) -- (m\i) ;
}
\end{scope}
\end{tikzpicture}
\caption{Contraction of vertices $u$ and $v$, and how the edges of the trigraph are updated.}
\label{fig:contraction}
\end{figure}

A \emph{$d$-sequence} (or \emph{contraction sequence}) is a sequence of \mbox{$d$-trigraphs} $G_n, G_{n-1}, \ldots, G_1$, where $G_n = G$, $G_1=K_1$ is the graph on a single vertex, and $G_{i-1}$ is obtained from $G_i$ by performing a single contraction of two (non-necessarily adjacent) vertices.
We observe that $G_i$ has precisely $i$ vertices, for every $i \in [n]$.
The twin-width of~$G$, denoted by $\tww(G)$, is the minimum integer~$d$ such that $G$ admits a~$d$-sequence.

For $u \in V(G_i)$, we denote by $u(G)$ the subset of $V(G)$ that was contracted to the single vertex $u$ in $G_n, G_{n-1}, \ldots, G_i$.
Twin-width and $d$-sequences can be equivalently seen as a partition refinement process on $V(G)$.
We start with the finest partition $\mathcal P_n = \{\{v\} : v \in V(G)\}$, and end with the coarsest partition $\mathcal P_1 = \{V(G)\}$.
There is a \emph{partition sequence} $\mathcal P_n, \mathcal P_{n-1}, \ldots, \mathcal P_2, \mathcal P_1$ mimicking the contraction sequence, where the contraction of $u, v \in V(G_i)$ corresponds to the merge of parts $u(G_i), v(G_i) \in \mathcal P_i$ to form the part $u(G_i) \cup v(G_i) = z(G_{i-1}) \in \mathcal P_{i-1}$, while all the other parts are unchanged from $P_i$ to $P_{i-1}$.
The red degree (bounded by~$d$) of a part $P \in \mathcal P_i$ now corresponds to the number of other parts $P' \in \mathcal P_i$ which are not fully adjacent nor fully non-adjacent to $P$ in $G$.
We may denote by $G_{\mathcal P}$ the trigraph corresponding to partition $\mathcal P$ over $V(G)$.
Thus $G_i = G_{\mathcal P_i}$.

\subsection{Classes with bounded twin-width and how the sequences are given}

The current paper is devoted to presenting efficient algorithms when the input has bounded twin-width, and the contraction sequence is given.
It is therefore important to know how realistic this scenario is.
Fortunately, in the first two papers of the series~\cite{twin-width1,twin-width2} we showed that many central (di)graph classes, be it sparse or dense, have bounded twin-width.
We summarize them here.
\begin{theorem}[\cite{twin-width1,twin-width2}]\label{thm:bd-tww}
  The following classes have bounded twin-width.
  \begin{compactitem}
  \item Bounded clique-width/rank-width, and more generally, boolean-width graphs,
  \item every hereditary proper subclass of permutation graphs,
  \item posets of bounded antichain size (seen as digraphs),
  \item unit interval graphs,\footnote{In this paper, we even show a linear-time algorithm finding a 2-sequence.}
  \item $K_t$-minor free graphs,
  \item map graphs,\footnote{To find the contraction sequence, we need to be given a map embedding.}
  \item subgraphs of $d$-dimensional grids,
  \item $K_t$-free unit $d$-dimensional ball graphs,
  \item $\Omega(\log n)$-subdivisions of all the $n$-vertex graphs,
  \item cubic expanders defined by iterative random 2-lifts\footnote{The actual definition of a 2-lift can be found in~\cite{twin-width2} but will not be needed here.} from $K_4$,\footnote{More generally, any graph built by successive $s$-lifts applied to $K_t$.}
  \item strong products of two bounded twin-width classes one of which has also bounded degree,
  \item any subgraph closure of a $K_{t,t}$-free bounded twin-width class, and
  \item any first-order interpretation\footnote{Actually a more general result is shown in the first paper of the series~\cite{twin-width1}.} of a bounded twin-width class.
  \end{compactitem}  
\end{theorem}  

Furthermore all our proofs are constructive and give rise to an $O(n^2)$-time algorithm to find an $O(1)$-sequence for an $n$-vertex graph of the class.
For some sparse classes, or dense classes with a sparse representation (like unit interval graphs), the sequence can even be found in quasi-linear time or even linear time.
Noticeably, we do \emph{not} know a polynomial-time algorithm that, given a ``general'' graph with bounded twin-width, outputs an $O(1)$-sequence.
Thus these algorithms are mostly ad hoc and specifically use properties of each listed class.
On the other hand, classes with unbounded twin-width include permutation graphs, cubic graphs, unit disk graphs, and $K_t$-free unit segment graphs.

It is striking that such a wide variety of seemingly unrelated graph classes allows for a unified algorithmic treatment.
One may think that this has to come with a prohibitive running time.
In fact our algorithms for \kmis and \kds run in the essentially optimal $2^{O(k)}n$-time (once the contraction sequence is computed), while our algorithms for \indsub and \subiso match the best known running time of $2^{O(k \log k)}n$ on $K_t$-minor free graphs.

It may seem surprising that, given the contraction sequence, our algorithms are linear (for fixed $k$) in the number of vertices, while the input graph $G$ may have $\Theta(n^2)$ edges.
Also the sequence itself consists of $n$ graphs on up to $n$ vertices, and the total number of \emph{vertices} in $G_n, \ldots, G_1$ is $\Theta(n^2)$.
The short answer is that we do not need to read the edges of $G$, nor all the vertices of all the trigraphs $G_i$.
Instead we only look, for every $i \in [n]$, at balls of radius\footnote{For \kds, the algorithm is more involved and this radius is function of $k$ and $d$.} $O(k)$ centered at the newly contracted vertex in the red graph of $G_i$.
Each such vertex set has size $d^{O(k)}$, so we may query red \emph{and} black edges within it.
The total number of operations remains bounded by $g(d,k) n$, for some function $g$.

One may still wonder if our algorithms can work with a compact encoding of the \mbox{$d$-sequence}, such as the mere list of contracted vertices.
The algorithms of~\cref{thm:bd-tww} computing the $d$-sequences all produce the union tree of how the vertices of $G$ are eventually merged into a single vertex.
Given this tree, we can solve the disjoint set problem (union-find) in optimal $O(n)$-time~\cite{Gabow85} (without inverse Ackermann function).
Thus we can, starting from $G$, perform the next contraction on the list, when the next trigraph of the sequence is needed.
The number of edge updates per contraction is a constant (more precisely $O(d)$).
One shall not forget, though, that we need in general $\omega(n)$-time to compute the sequence in the first place.  

\section{Practical algorithms for \kmis and its generalizations}\label{sec:kmis}

In this section, we present essentially optimal fixed-parameter algorithms for \kmis, \indsub, \subiso, on graphs of bounded twin-width.
The crux for the running time analysis is a simple bound on the number of connected subsets of size at most $k$ in a bounded-degree graph.
The key to show this folklore lemma is that a connected subgraph of size at most $k$ can be spanned by a walk of length at most $2k-3$.

\begin{lemma}[folklore]\label{lem:connected-subgraphs}
  The number of vertex subsets of size at most $k$ inducing a connected subgraph in an $n$-vertex graph of maximum degree $d$ is at most $(d^{2k-2}+1) n$. 
\end{lemma}
\begin{proof}
  If $d=0$ or $d=1$, the total number of connected subgraphs is $n$ or at most $3n/2$, respectively.
  Thus the claim holds in these cases, and we now assume that $d \geqslant 2$.
  Every connected subgraph $H$ has a spanning tree, say, $T_H$ rooted at $v_H$.
  The circumnavigation of $T_H$ from $v_H$ follows every edge of $T_H$ at most twice.
  Moreover if we only span $T_H$ without going back to $v_H$ in the end, at least one edge of $T_H$ is taken only once.
  Hence every connected subgraph of size at most $k$ can be described by a starting vertex ($n$ choices) followed by a walk on $2k-3$ other vertices (at most $d$ choices for each).
  Therefore the number of connected vertex subsets of size at most $k$ is bounded by $n \Sigma_{0 \leqslant i \leqslant 2k-3} d^i \leqslant n d^{2k-2}$.
\end{proof}

We get the following as a direct corollary of the previous proof.
\begin{corollary}\label{cor:connected-subgraphs}
  The number of connected vertex sets of size at most $k$, intersecting a set~$X$, in a graph of maximum degree $d$ is at most $(d^{2k-2}+1) |X|$.
  Furthermore they can be enumerated in time $O(d^{2k-2}|X|)$.
\end{corollary}

We now show how to solve \kmis by dynamic programming on the connected subsets of size at most $k$ in the red graphs of a $d$-sequence given with the input graph. 

\begin{theorem}\label{thm:k-mis}
  Given an $n$-vertex graph $G$, a positive integer $k$, and a $d$-sequence $G=G_n, \ldots, G_1=K_1$, \kmis can be solved in time $O(k^2 d^{2k}n)=2^{O_d(k)}n$.
\end{theorem}
\begin{proof}
  Our algorithm maintains a set of \emph{optimum partial solutions} in the current trigraph, starting from $G$, and progressively going along the $d$-sequence.
  Let us start with a definition of the partial solutions and of their optimality.
  
  A~\emph{partial solution} in the trigraph $G_i$ is a pair $(T,S)$ where $T \subseteq V(G_i)$ is a vertex set inducing a connected subgraph in the red graph $(V(G_i),R(G_i))$, and $S \subseteq V(G)$ is an independent set of $G$ such that $S \subseteq \bigcup_{u \in T} u(G)$ and for every $u \in T$, $S \cap u(G) \neq \emptyset$.
  A~partial solution $(T,S)$ is said \emph{optimum} if there is no partial solution $(T,S')$ such that $|S| < |S'|$.
  A~set $T \subseteq V(G_i)$ is said \emph{realizable} (in $G_i$) if there is an $S \subseteq V(G)$ such that $(T,S)$ is a partial solution in $G_i$.
  Notice that \emph{not} every connected subset in the red graph is realizable.
  For instance, it is easy to engineer a situation where there is no independent set intersecting the three vertices of a 3-vertex red path.
  Initially, in $G$, the only connected subgraphs of the red graph are singletons (since there is no red edge).
  So there are exactly $n$ (optimum) partial solutions in $G=G_n$: Each vertex $v$ of $G$ induces a partial solution $(\{v\},\{v\})$.
  We denote by $\mathcal S_n$ this set of $n$ optimum partial solutions.
  It boils down to determining if there is a partial solution $(\_,S)$ in $G_1$ (or actually in any $G_i$) with $|S| \geqslant k$.
  For $i$ going from $n-1$ down to~1, we will build a set of optimum partial solutions $\mathcal S_i$ in $G_i$ from the set $\mathcal S_{i+1}$, keeping the invariant that for every realizable set $T \subseteq V(G_i)$, there is a unique optimum partial solution $(T,S)$ stored in $\mathcal S_i$ (and no other partial solution in $\mathcal S_i$).

  We shall then describe how we update the set of optimum partial solutions after a single contraction.
  Two partial solutions $(T,\_)$ and $(T',\_)$ in $G_i$ are said \emph{disjoint} if $T \cap T' = \emptyset$, and \emph{separate}, if they are disjoint and there is no red edge $uu' \in R(G_i)$ with $u \in T$ and $u' \in T'$.
  Two separate partial solutions $(T,\_)$ and $(T',\_)$ are said \emph{compatible} if there is no edge $uu' \in E(G_i) \cup R(G_i)$ with $u \in T$ and $u' \in T'$.
  The \emph{union} of two compatible partial solutions $(T_1,S_1)$ and $(T_2,S_2)$ as $(T_1,S_1) \cup (T_2,S_2) := (T_1 \cup T_2,S_1 \cup S_2)$.
  By definition, such a union is \emph{not} a partial solution since $T$ induces two connected components in its current red graph. 
  Nevertheless we will build the new (connected) partial solutions of $G_i$ by making unions of up to $d+2$ pairwise compatible partial solutions in $G_{i+1}$.
  These unions will be connected in $G_i$, hence will correspond to partial solutions as well.

  Let us be more specific.
  Say $u, v \in V(G_{i+1})$ are contracted into $z \in V(G_i)$ to form $G_i$.
  We say that a partial solution $(T,\_)$ in $G_i$ \emph{intersects} a set $X \subseteq V(G_i)$ if $T \cap X \neq \emptyset$. 
  We initialize $\mathcal S_i$ with all the partial solutions of $\mathcal S_{i+1}$ not intersecting $\{u,v\}$.
  We now add one partial solution in $\mathcal S_i$ per realizable set $T \ni z$ in $G_i$, of size at most $k$.
  For every $T \subseteq V(G_i)$ such that $z \in T$ and $T$ induces a connected subgraph on at most $k$ vertices in the red graph $(V(G_i),R(G_i))$, we observe three possibilities for a potential partial solution $(T,S)$.
  Either $S$ intersects $u(G)$ and $v(G)$, or it intersects only $u(G)$, or it intersects only $v(G)$.
  (It is not possible that $S \cap (u(G) \cup v(G)) = \emptyset$ since $T$ contains $z$.)
  Therefore we take the best (meaning with the largest $S$, breaking ties arbitrarily) of the potential partial solutions $\bigcup \dec(T \setminus \{z\} \cup \{u, v\}), \bigcup \dec(T \setminus \{z\} \cup \{u\}),\bigcup \dec(T \setminus \{z\} \cup \{v\})$, where $\dec(X)$ is the set with one partial solution per connected component of $X$ in its red graph (here $(V(G_{i+1}),R(G_{i+1})$).
  See \cref{fig:k-is-update} for an illustration of this decomposition.
  In the very possible event that at least one such connected component of $X$ is not realizable, $\dec(X) =$ None.
  The union $\bigcup \dec(X)$ of all the partial solutions of $\dec(X)$ is None if $\dec(X)=$ None or if there is at least one black edge between two connected components.
  Otherwise $\bigcup \dec(X)$ is a pair $(T,S)$ as defined in the previous paragraph, since the partial solutions of $\dec(X)$ are pairwise compatible.
  Since $T$ is chosen connected in $(V(G_i),R(G_i))$, $(T,S)$ is indeed a partial solution in $G_i$.
  If $\bigcup \dec(T \setminus \{z\} \cup \{u, v\}), \bigcup \dec(T \setminus \{z\} \cup \{u\}),\bigcup \dec(T \setminus \{z\} \cup \{v\})$ all three evaluate to None, then $\best\{\bigcup \dec(T \setminus \{z\} \cup \{u, v\}), \bigcup \dec(T \setminus \{z\} \cup \{u\}),\bigcup \dec(T \setminus \{z\} \cup \{v\})\}$ also returns None.
  This would mean that $T$ is not realizable.
  If instead $T$ is realizable, we get a partial solution $(T,S)$ that we put in $\mathcal S_i$.
  If $|S| \geqslant k$, we already have a large enough independent set; the algorithm outputs it and terminates.

  If we finally build $\mathcal S_1$, and no independent set of size at least $k$ was found, we output $S$, the unique set such that $(\_,S) \in \mathcal S_1$.
  $\mathcal S_1$ is indeed a singleton since there is only one realizable set in $G_1$.
  That finishes the description of the algorithm \algkmis, see Algorithm~\ref{alg:kmis}.

  \medskip
  
\begin{algorithm}
  \DontPrintSemicolon
  \SetKwInOut{Input}{Input}
  \SetKwInOut{Output}{Output}
  \Input{~~A graph $G$, a positive integer $k$, and a $d$-sequence $G=G_n, \ldots, G_1=K_1$.}
  \Output{~~An independent set of $G$ of size at least $\min(k,\alpha(G))$.}
  $\mathcal S_n \leftarrow \bigcup_{v \in V(G)} \{(\{v\},\{v\})\}$\;
 \For{$i = n-1 \rightarrow 1$}{
   $u, v \leftarrow $ contracted pair in $G_{i+1} \to G_i$\;
   $z \leftarrow $ contraction of $u$ and $v$ in $G_i$\;
   $\mathcal S_i \leftarrow$ partial solutions of $\mathcal S_{i+1}$ \emph{not} intersecting $\{u,v\}$\;
  \For{every vertex subset $T$ connected in $(V(G_i),R(G_i))$, with $z \in T$ and $|T| \leqslant k$}{
    $(T,S) \leftarrow~\best\{\bigcup \dec(T \setminus \{z\} \cup \{u, v\}), \bigcup \dec(T \setminus \{z\} \cup \{u\}),\bigcup \dec(T \setminus \{z\} \cup \{v\})\}$\;
    \If{$|S| \geqslant k$}{
      \Return{$S$}\;
    }
    \If{$(T,S) \neq$ None}{
      $\mathcal S_i \leftarrow \mathcal S_i \cup \{(T,S)\}$\;
    }
  } 
 }
 $\{(S,\_)\} \leftarrow \mathcal S_1$\; 
 \Return{$S$}\;
 \caption{\algkmis}
 \label{alg:kmis}
\end{algorithm}

\begin{figure}[h!]
  \centering
  \begin{tikzpicture}
    \def\s{0.9}
    \def\r{0.16}
    \node at (\s * -1, \s * 6) {$G_i$} ;
     \foreach \i/\j [count = \k from 3] in {4/1,4/2,4/3,4/4,5/0,5/2,5/3,5/4/5/5,6/2,6/4, 0/1,0/2.5,0/4,-1/1,-1/2,-1/3,-1/4, 1/0,2/0,3/0,0/-1,1/-1,2/-1,3/-1,4/-1, 1/5,2/5,3/5,0/6,1.5/6,3/6,4.5/6,0/7,1.5/7,3/7,4.5/7}{
      \node[draw,circle] (a\k) at (\s * \i,\s * \j) {} ;
     }
     \node[draw,circle] (az) at (\s * 2, \s * 2.5) {} ;
     \node at (\s * 1.6, \s * 2.6) {$z$} ;
    \foreach \i/\j in {z/13,3/7,3/4,3/8, 4/9,10/12,12/9,31/36,30/34,21/25,21/26,14/18,14/17,20/13,16/23,7/22,7/27,6/34,31/15}{
      \draw (a\i) -- (a\j) ;
    }
    \foreach \i/\j in {z/3,z/4,z/5,z/6,z/20,z/21,z/22,z/14,z/15, z/28,z/29,z/30, 4/8,8/11,8/9,9/5, 6/10, 28/31,31/35,28/32,31/32,32/36,36/35,30/33,33/37,37/38,38/34,34/33, 20/23,24/23,20/24,24/25,20/25,22/26,22/27,15/19,19/18,18/17,17/16,16/13,13/17,10/34,31/19}{
      \draw[red,thick] (a\i) -- (a\j) ;
    }
    \foreach \i/\c in {z,28,31,19,18,17,29,30,32,33,37,38,6,10,5,9,3,20,21,23,24}{
      \fill[opacity=0.5] (a\i) circle (\r * \s);
    }

    \begin{scope}[xshift=-7.5cm]
      \node at (\s * -1, \s * 6) {$G_{i+1}$} ;
    \foreach \i/\j [count = \k] in {2/2,2/3,4/1,4/2,4/3,4/4,5/0,5/2,5/3,5/4/5/5,6/2,6/4, 0/1,0/2.5,0/4,-1/1,-1/2,-1/3,-1/4, 1/0,2/0,3/0,0/-1,1/-1,2/-1,3/-1,4/-1, 1/5,2/5,3/5,0/6,1.5/6,3/6,4.5/6,0/7,1.5/7,3/7,4.5/7}{
      \node[draw,circle] (a\k) at (\s * \i,\s * \j) {} ;
    }
    \node at (\s * 2, \s * 2) {$u$} ;
    \node at (\s * 2, \s * 2.68) {$v$} ;
    \node[draw,rounded corners,dashed,fit=(a1) (a2)] {} ;
    \foreach \i/\j in {1/3,1/4,1/5,1/6,1/13,1/30,2/13,2/14,2/15, 3/7,3/4,3/8, 4/9,10/12,12/9,31/36,30/34,21/25,21/26,14/18,14/17,20/13,16/23,7/22,7/27,6/34,31/15}{
      \draw (a\i) -- (a\j) ;
    }
    \foreach \i/\j in {1/20,1/21,1/22, 2/28,2/29,2/30, 4/8,8/11,8/9,9/5, 6/10, 28/31,31/35,28/32,31/32,32/36,36/35,30/33,33/37,37/38,38/34,34/33, 20/23,24/23,20/24,24/25,20/25,22/26,22/27,15/19,19/18,18/17,17/16,16/13,13/17,10/34,31/19}{
      \draw[red,thick] (a\i) -- (a\j) ;
    }

    \foreach \i/\c in {2/blue,28/blue,31/blue,19/blue,18/blue,17/blue,29/blue,30/blue,32/blue,33/blue,37/blue,38/blue, 6/orange,10/orange, 5/green!40!black,9/green!40!black, 3/cyan, 21/purple, 20/yellow!80!black, 23/yellow!80!black, 24/yellow!80!black}{
      \fill[\c,opacity=0.8] (a\i) circle (\r * \s);
    }
    \end{scope}
  \end{tikzpicture}
  \caption{Right: In gray, a connected vertex set $T$ in the red graph of $G_i$ in the vicinity of the just contracted vertex $z \in T$. Left: The decomposition $\dec(T \setminus \{z\} \cup  \{v\})$ in the previous trigraph $G_{i+1}$, where each color represents a connected component. If every color class is a realizable set in $G_{i+1}$, then $T$ is realizable in~$G_i$, with (optimum) partial solution $\bigcup \dec(T \setminus \{z\} \cup  \{v\})$. Note that, due to black edges between~$u$ and some vertices of $T$, the partial solutions in $\dec(T \setminus \{z\} \cup  \{u,v\})$ and in $\dec(T \setminus \{z\} \cup  \{u\})$ cannot be pairwise compatible.}
  \label{fig:k-is-update}
\end{figure}
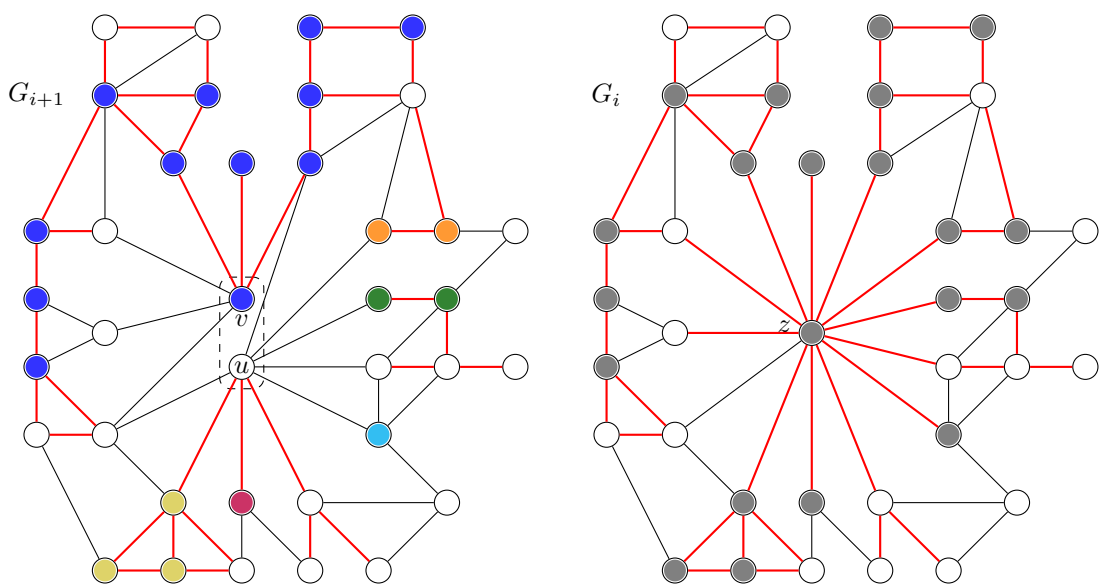

\medskip

\textbf{Correctness.}
By a transparent induction, any set returned by \algkmis is an independent set.
Indeed the initial partial solutions (in $\mathcal S_n$) are singletons.
Every new partial solution is formed by taking a union of independent sets such that there is no black or red edge between any pair of independent sets.
Hence the union is overall an independent set.

We now claim that if there is an independent set of size at least $k$ in $G$, then \algkmis indeed outputs a solution of size at least $k$.
Again we show by induction the following invariant: For every realizable set $T \subseteq V(G_i)$ (in $G_i$) of size at most $k$, $\mathcal S_i$ (eventually) contains a solution $(T,S)$ such that $|S| = \alpha(G[\bigcup_{u \in T} u(G)])$ or $|S| \geqslant k$.
The former condition, ``$|S| = \alpha(G[\bigcup_{u \in T} u(G)])$'', is initially true for the singletons of $\mathcal S_n$.
If the latter condition, ``$|S| \geqslant k$'', ever happens, \algkmis outputs it and we are done.
Thus for the induction hypothesis of $\mathcal S_{i+1}$, we suppose that the former condition always holds.

Say, $u, v \in V(G_{i+1})$ are contracted into $z \in V(G_i)$.
Let $T$ be a realizable set in $G_i$.
If $z \notin T$, then $T$ is also a realizable set in $G_{i+1}$.
By the induction hypothesis, there is a partial solution $(T,S^*)$ in $\mathcal S_{i+1}$ such that $|S^*| = \alpha(G[\bigcup_{u \in T} u(G)])$.
This partial solution was simply transmitted from $\mathcal S_{i+1}$ to $\mathcal S_i$, hence $(T,S^*) \in \mathcal S_i$.

Let us now assume that $z \in T$.
We fix $S'$, a maximum independent set in $G[\bigcup_{u \in T} u(G)]$.
The algorithm \algkmis defines the partial solution $(T,S) \in \mathcal S_i$ by taking the best of the at most three unions $\bigcup \dec(T \setminus \{z\} \cup \{u,v\})$, $\bigcup \dec(T \setminus \{z\} \cup \{u\})$, and $\bigcup \dec(T \setminus \{z\} \cup \{v\})$ (note that at most two of those may not be defined). 
Build the set $\emptyset \neq I \subseteq \{u,v\}$ by putting $u$ (resp.~$v$) in $I$ if $S' \cap u(G) \neq \emptyset$ (resp.~$S' \cap v(G) \neq \emptyset$).
We consider $\dec(T \setminus \{z\} \cup I)$, the partial solutions in $\mathcal S_{i+1}$ associated to each connected component of $T \setminus \{z\} \cup I$ in $(V(G_{i+1}),E(G_{i+1}) \cup R(G_{i+1}))$ (by the existence of $S'$, each such connected component is indeed realizable).
By the induction hypothesis, every partial solution of $\dec(T \setminus \{z\} \cup I)$ is optimum.
Thus the union $\bigcup \dec(T \setminus \{z\} \cup I)$ has the same size as $S'$.
This implies that the partial solution $(T,S)$ put in $\mathcal S_i$ is also optimum.

Finally if \algkmis terminates without reporting an independent set of size at least $k$, our invariant on $\mathcal S_1$ indicates that $\alpha(G) < k$.
In that case the unique (optimum) partial solution $(V(K_1),S) \in \mathcal S_1$ verifies $|S| = \alpha(G)$.

\medskip

\textbf{Running time.}
The claimed running time for \algkmis essentially relies on~\cref{cor:connected-subgraphs}.
By this corollary, the sets $T$ of the inner for loop (line 6) can be enumerated in time $O(d^{2k})$.
The connected components of line 7 can be computed in time $O(\min(d,k)k)$, say, by breadth-first search in the red graph of $G_i$.
Then checking the absence of black edges between potential partial solutions takes time $O(k^2)$.
Thus the overall running time is $O(k^2 d^{2k} n)$.
Interestingly, once the trigraphs of a $d$-sequence of $G$ have been computed, \kmis can be solved in sublinear time in the size of $G$, when $k^2 d^{2k}n = o(|E(G)|)$.
Another observation is that when the twin-width $d$ is polylogarithmic in $n$, i.e., in $\Theta(\log^c n)$, \algkmis is still fixed-parameter tractable in $k$.
Indeed $\log^{O(k)} n = k^{O(k)} n$ as noticed by Sloper and Telle~\cite{Sloper08}, which implies that \algkmis runs in time $2^{O(k \log k)} n^2$ in that regime. 

\medskip

\textbf{Optimizations.}
We suggest some improvements or variations of \algkmis to generally improve over the worst-case running time of the inner for loop.
A lot of sets $T$ will trivially be \emph{not} realizable because they induce a black edge.
When enumerating the walks starting at $z$ of length at most $2k-3$, one can abort every branch $z v_1 \ldots v_h$ inducing at least one black edge.
It can even be done in a way that the enumeration takes time $O(t)$ where $t$ is the number of sets $T \ni z$ of size at most $k$, such that $T$ is connected in the red graph, and an independent set in the black graph.

Even if a set $T$ satisfies those properties, we have no guarantee that $T$ is realizable.
In very dense instances, it is imaginable that the realizable sets are very rare.
In that case, we will lose a lot of time generating sets $T$ to observe immediately after that there is no associated partial solution $(T,S)$.
An alternative to \algkmis is to build the new partial solutions of $\mathcal S_i$ directly as unions of pairwise compatible partial solutions of $\mathcal S_{i+1}$, without anticipating the nature of the possibly realizable set $T \subseteq V(G)$.

Let us be more precise.
Let $R_z$ be the set of red neighbors of $z$ in $G_i$.
For every set of at most $\max(2,d+1)$ partial solutions $(T_1,S_1), \ldots, (T_h,S_h) \in \mathcal S_{i+1}$ intersecting $R_z$, at least one of which intersects $\{u,v\}$, if the partial solutions are pairwise compatible, we update the realizable set $\bigcup_{i \in [h]} T_i$ with the partial solution $\bigcup_{i \in [h]} (T_i,S_i)$ if $\bigcup_{i \in [h]} S_i$ is larger than the current best solution.
Following the first improvement, we can only generate the sets that are pairwise compatible.
As we know, there are at most three ways to reach a given set $T \subseteq V(G_i)$ as a union of pairwise compatible partial solutions in $\mathcal S_{i+1}$.
The running time of this variation of \algkmis is $O^*(\Sigma_{i \in [n]} |\mathcal S^{\text{new}}_i|)$, where $S^{\text{new}}_i := \mathcal S_i \setminus \mathcal S_{i-1}$ (and $S^{\text{new}}_n := \mathcal S_n$) represents the new partial solutions computed at step $i$.
In practice, this can be significantly better than $O(k^2 d^{2k}n)$.
Such a dynamic programming, only generating ``positive'' subinstances, dubbed \emph{positive-instance driven} by Tamaki, led to a breakthrough and current state-of-the-art practical algorithm for computing optimally the treewidth of a graph~\cite{Tamaki19}.

\medskip

\textbf{Weights.}
Without too many changes, \algkmis may support weights, that is, find an independent set of size exactly $\min(k,\alpha(G))$ with largest total weight.
Instead of keeping one solution $S$ per realizable set $T$, we keep up to $k$ solutions, one per pair $(T,j)$ with $j \in [|T|,k]$.
A~partial solution $(T,j,S)$ is defined as before except $S$ is required to have size exactly $j$.
To compute the new partial solutions, we add a third nested for loop after line 6: We iterate over all the ways of distributing $j \leqslant k$ units between the red connected components induced by $T' \in \{T \setminus \{z\} \cup \{u,v\}, T \setminus \{z\} \cup \{u\}, T \setminus \{z\} \cup \{v\}\}$ so that each connected component gets a positive integer (at least equal to its size).
We then add to $\mathcal S_i$ one partial solution $(T,j,S)$ (if at least one exists) maximizing the weight of $S$ for fixed $T$ and $j$.
We also skip lines 8 and 9 of \algkmis.

This comes with a slight increase in the running time.
Namely, there is an extra $2^{O(k \log k)}$ factor accounting for the ordered partition of integer $j \leqslant k$ into positive integers.
Thus the overall running time with weights is $2^{O(k \log k)}d^{2k}n$.
\end{proof}

As twin-width and $d$-sequences are preserved when complementing the graph, we also solve \clique in the same running time.
One may wonder if the dependency in $k$ of our $2^{O_d(k)}n$-time algorithm can be improved.
It turns out that this running time is essentially optimal.
Due to the Sparsification Lemma~\cite{sparsification} and folklore reductions, \mis restricted to subcubic $n$-vertex graphs cannot be solved in $2^{o(n)}$, under the Exponential Time Hypothesis\footnote{The assumption that there is a constant $\delta>0$, such that \textsc{3-SAT} cannot be solved in time $2^{\delta n}$.} (ETH)~\cite{Impagliazzo01}.
Thus, by the classic self-reduction consisting of performing an even subdivision of each edge \cite{Poljak74}, \mis cannot be solved in time $2^{o(n/\log n)}$ on $2 \lceil \log n \rceil$-subdivisions of $n$-vertex subcubic graphs, unless the ETH fails.
In~\cite{twin-width2}, we show how to find $O(1)$-sequences in polynomial time for $2 \lceil \log n \rceil$-subdivisions of $n$-vertex graphs.
Therefore this lower bound holds even if we are given the $d$-sequence.
In particular, no algorithm solves \kmis in time $2^{o_d(k / \log k)}n^{O(1)}$, unless the ETH fails.

If $\mathcal T$ is a $d$-sequence $G=G_n, \ldots, G_1=K_1$, we denote by $\mathcal C_{\mathcal T}$ denote the set of connected vertex subsets in a red graph of some trigraph $G_i \in \mathcal T$.
Let us also denote by $\mathcal C_{\mathcal T, k}$ the set of connected vertex subsets of size at most $k$ in a red graph of some trigraph $G_i \in \mathcal T$.
In both cases, the exact same vertex subset appearing connected in several trigraphs of $\mathcal T$ counts only once.
We know that $|\mathcal C_{\mathcal T, k}| \leqslant d^{2k}n$ but, as we already observed, $|\mathcal C_{\mathcal T, k}|$ can in principle be much smaller.
As a consequence of our proof of \cref{thm:k-mis}, we obtain the following.

\begin{theorem}\label{thm:ctk-ct}
Given as input an $n$-vertex graph $G$ and a $d$-sequence $G=G_n, \ldots, G_1=K_1$, \kmis can be solved in time $O^*(|\mathcal C_{\mathcal T, k}|)$ and \lmis can be solved in time $O^*(|\mathcal C_{\mathcal T}|)$.
\end{theorem}

We actually showed the stronger result that \kmis and \lmis can be solved in time $O^*(|\mathcal R_{\mathcal T, k}|)$ and $O^*(|\mathcal R_{\mathcal T}|)$, respectively, where $R_{\mathcal T, k} \subseteq \mathcal C_{\mathcal T, k}$ and $R_{\mathcal T} \subseteq \mathcal C_{\mathcal T}$ only consist of the realizable sets.
In~\cite{twin-width1}, we show how to find in polynomial time $f(\text{rw})$-sequences for $n$-vertex graphs with rank-width (even boolean-width) at most~$\text{rw}$. 
Importantly the sequences comprise only $g(\text{rw})n$ connected vertex subsets.
Hence \cref{thm:ctk-ct} in particular generalizes the $O(n)$-time algorithm for \mis in graphs of bounded rank-width/clique-width, given the rank- or clique-decomposition.
Indeed the polynomial algorithm computing the $f(\text{rw})$-sequence takes time $O(n)$, provided the rank-width decomposition.
Of course \cref{thm:ctk-ct} is more general than that.
In light of the next corollary, it also yields a polynomial-time algorithm when a 2-sequence can be efficiently computed.

\begin{corollary}\label{cor:tww2}
  Given as input an $n$-vertex graph $G$ and a 2-sequence $G=G_n, \ldots, G_1=K_1$, \lmis can be solved in polynomial time.
\end{corollary}

\begin{proof}
  The red graphs of the trigraphs of the 2-sequence $\mathcal T = G_n, \ldots, G_1$ are disjoint unions of paths and cycles (their degree is at most~2).
  Thus each $(V(G_i),R(G_i))$ has at most $n^2$ connected vertex subsets. 
  Hence $|\mathcal C_{\mathcal T}| = O(n^3)$.
  We conclude by~\cref{thm:ctk-ct}.
\end{proof}

As we will now see, \cref{cor:tww2} captures unit interval graphs, which have unbounded rank-width.
\begin{lemma}\label{lem:unitinterval}
Unit interval graphs have twin-width 2.
\end{lemma}
\begin{proof}
  Consider the unit interval graph $I_{k,nk}$ on vertex set $[nk]$ where, for every $j \in [nk]$, the interval of length exactly $k$ and with left endpoint $j$ is present.
  The family $I_{k,nk}$ is universal in the sense that every unit interval graph is an induced subgraph of some $I_{k,nk}$.
  For every $i \in [n]$, contract $ki-1$ and $ki$.
  Then for every $i \in [n]$ in increasing order, contract $ki-2$ with $\{ki-1,ki\}$, etc.
  At every stage, the only red edges are between two consecutive contracted groups, forming a path.
  We eventually end up with only a red path, which has twin-width 2.
\end{proof}


We now extend~\cref{thm:k-mis} in two directions.
We show that \textsc{(Induced) Subgraph Isomorphism} and \scaset can be solved in time $2^{O(k \log k)}n$ on graphs given with an $O(1)$-contraction sequence.

\begin{theorem}\label{thm:sub-iso}
  Given a graph $G$, a $d$-sequence $G=G_n, G_{n-1},$ $\ldots, G_1=K_1$, and a pattern graph $H$ on $k$ vertices, \subiso and \indsub can be solved in time $2^{O(k \log k)}d^{2k}n=2^{O_d(k \log k)}n$.
\end{theorem}
\begin{proof}
  The algorithms are almost identical and are obtained by making some additions and modifications to \algkmis.
  We will first describe the algorithm \algindsub for \indsub.
  The algorithm \algsubiso solving \subiso will be obtained by changing a single word in the pseudo-code (see Algorithm~\ref{alg:indsub}).
  
  We identify $V(H)$ to the set of integers $[k]$.
  A \emph{division} of $T \subseteq V(G_i)$\footnote{In this definition, we do \emph{not} require that $T$ is connected in the red graph.} is a mapping $\eta$ from $T$ to $2^{[k]} \setminus \{\emptyset\}$ such that $\eta(u) \cap \eta(v) = \emptyset$ for every $u \neq v \in T$.
  We define $\eta(T)$ as $\bigcup_{u \in T} \eta(u)$.
  Given a realizable set $T \subseteq V(G_i)$ and a division $\eta$ of $T$, a set $S \subseteq V(G)$ is said $(T,\eta)$-compliant (or simply compliant, if $T$ and $\eta$ are clear from the context) if there is an induced subgraph isomorphism $\lambda$ from $H[\eta(T)]$ to $G[S]$, such that $S \cap u(G) = \lambda(\eta(u))$ for every $u \in T$.
  Now partial solutions in $G_i$ are triples $(T,\eta,S)$ where $T \subseteq V(G_i)$ is still a vertex set of size at most $k$ inducing a connected subgraph in $(V(G_i),R(G_i))$, $\eta$ is a division of $T$, and $S \subseteq V(G)$ is $(T,\eta)$-compliant.
  In particular $S \subseteq \bigcup_{u \in T} u(G)$ and $S \cap u(G) \neq \emptyset$, as it was the case for \kmis.

  It is simpler to first present the new algorithms with a classic (static) dynamic programming.
  As before this can be turned into its ``positive-instance driven'' version.
  We maintain a table $\mathcal T$, where for every realizable set $T \subseteq V(G_i)$ and every division $\eta$ of $T$, $\mathcal T[T,\eta]$ is intended to contain a $(T,\eta)$-compliant set $S \subseteq V(G)$ if it exists, and ``None'' otherwise.
  It can be observed that for every vertex $v \in V(G)$, the singleton $\{v\}$ is $(\{v\},\eta)$-compliant for every division $\eta$ of $\{v\}$.
  Notice that a division of $\{v\}$ assigns a single vertex $j \in V(H) = [k]$ to $v$.
  We therefore initialize $\mathcal T$ by putting $\{v\}$ in each cell $\mathcal T[\{v\},\eta: v \mapsto \{j\}]$, for every $v \in V(G)$ and $j \in [k]$.
  By default, if a cell of $\mathcal T$ is not filled, it contains the value ``None''.

  As in the algorithm of \cref{thm:k-mis}, we can compute the partial solutions in $G_i$ from the partial solutions in $G_{i+1}$.
  Say that to go from $G_{i+1}$ to $G_i$, we contract $u, v \in V(G_{i+1})$ into $z \in V(G_i)$.
  Note that every cell $\mathcal T[T,\_]$ such that $T \subseteq V(G_i) \setminus \{z\}$ was previously filled.
  Indeed a set $T \subseteq V(G_i) \setminus \{z\}$ connected in $(V(G_i),R(G_i))$ is also connected in $(V(G_{i+1}),R(G_{i+1}))$ (and included in $V(G_{i+1}) \setminus \{u,v\}$).
  We shall fill the cells $\mathcal T[T,\_]$ such that $z \in T \subseteq V(G_i)$.
  Again we build these partial solutions as union of partial solutions in $G_{i+1}$.
  The fact $z \in T$ entails that such a union may cover $u$, or $v$, or both.
  For every $I \in \{\{u\},\{v\},\{u,v\}\}$, we decompose $T' := T \setminus \{z\} \cup I$ into its connected component $T_1, \ldots, T_h$ in the red graph $(V(G_{i+1}),R(G_{i+1}))$.
  Any division $\eta$ of $T'$ naturally breaks into $h$ divisions $\eta_1, \ldots, \eta_h$ where $\eta_p$ is a division of $T_p$ for every $p \in [h]$.
  We denote by $\dec(T',\eta)$ the $h$ pairs $(T_1,\eta_1), \ldots, (T_h,\eta_h)$.
  
  For every such pair $(T',\eta)$, we fill $\mathcal T[T',\eta]$ with an actual solution if the following holds.
  First, every entry $\mathcal T[T_p,\eta_p]$, for $p \in [h]$, should contain an actual solution $S_p$ (which is not ``None'').
  Secondly, for every $p \neq p' \in [h]$ the edges and non-edges in $H$ between $\eta_p(T_p)$ and $\eta_{p'}(T_{p'})$ should match the edges and non-edges in $G$ between $S_p$ and $S_{p'}$.
  More precisely, there should be a bijection $\lambda$ from $\eta_p(T_p) \cup \eta_{p'}(T_{p'})$ to $S_p \cup S_{p'}$ such that $\lambda(\eta(x)) = (S_p \cup S_{p'}) \cap x(G)$ for every $x \in T_p \cup T_{p'}$ where $\eta(x) := \eta_p(x)$ if $x \in T_p$ and $\eta(x) := \eta_{p'}(x)$ if $x \in T_{p'}$, and $ab \in E_H(\eta_p(T_p), \eta_{p'}(T_{p'}))$ if and only if $\lambda(a)\lambda(b) \in E_G(S_p,S_{p'})$.
  Such a bijection $\lambda$ is called an \emph{$(\eta_p,\eta_{p'})$-isomorphism}.
  We also say that $H[\eta_p(T_p), \eta_{p'}(T_{p'})]$ is \emph{$(\eta_p,\eta_{p'})$-isomorphic} to $G[S_p, S_{p'}]$.
  Since $T_p$ and $T_{p'}$ induce two connected components in the red graph of $G_{i+1}$, there are only black edges and non-edges between pairs $x \in T_p, x' \in T_{p'}$.
  Thus the notion of $(\eta_p,\eta_{p'})$-isomorphism crucially does \emph{not} depend on $S_p$ and $S_{p'}$: If $ab \in E_H(\eta_p(T_p), \eta_{p'}(T_{p'}))$ (resp.~$ab \notin E_H(\eta_p(T_p), \eta_{p'}(T_{p'}))$), we check that there is a black edge (resp.~a non-edge) between $x \in T_p$ and $y \in T_{p'}$ where $x$ and $y$ are the only vertices in $T_p \cup T_{p'}$ such that $a \in \eta_p(x)$ and $b \in \eta_{p'}(y)$.
  If both conditions of this paragraph are fulfilled, we put $\bigcup_{p \in [h]} S_i$ in cell $\mathcal T[T',\eta]$ (otherwise the content of this cell remains unchanged).

  If we ever fill a cell $\mathcal T[T',\eta]$ where $\eta(T') = [k]$ with an actual solution $S$, \algindsub reports~$S$ as an overall solution of the \indsub-instance.
  If after all the partial solutions in $G_1$ are computed (i.e., after we exit the outermost for loop in Algorithm~\ref{alg:indsub}), no such solution was reported, \algindsub outputs that no solution exists.
  This terminates the description of \algindsub.
  For \algsubiso, we just replace the occurrences of ``induced subgraph'' by ``subgraph''.
  In the definition of the partial solutions, the mapping $\lambda$ is now a (non-induced) subgraph isomorphism from $H[\eta(T)]$ to $G[S]$.
  In the update of the partial solutions, we also relax the $(\eta_p,\eta_{p'})$-isomorphism to be a mere \emph{$(\eta_p,\eta_{p'})$-subisomorphism} preserving the edges of $H$, but not necessarily its non-edges.
  See Algorithm~\ref{alg:indsub} for the pseudo-code of both algorithms.

  \medskip
  
  \begin{algorithm}
    \DontPrintSemicolon
  \SetKwInOut{Input}{Input}
  \SetKwInOut{Output}{Output}
  \Input{~~A graph $G$, a $d$-sequence $G=G_n, \ldots, G_1=K_1$, and a graph $H$ on $[k]$.}
  \Output{~~A set $S$ such that $G[S]$ and $H$ are isomorphic, if it exists.}
  \For{$v \in V(G)$}{
    \For{$j = 1 \rightarrow k$}{
      $\mathcal T[\{v\},\eta: v \mapsto \{j\}] \leftarrow \{v\}$\;
    }
  }
 \For{$i = n-1 \rightarrow 1$}{
   $u, v \leftarrow $ contracted pair in $G_{i+1} \to G_i$\;
   $z \leftarrow $ contraction of $u$ and $v$ in $G_i$\;
   \For{every vertex subset $T$ connected in $(V(G_i),R(G_i))$, with $z \in T$ and $|T| \leqslant k$}{
     \For{$I \in \{\{u,v\},\{u\},\{v\}\}$}{
       \For{every division $\eta$ of $T \setminus \{z\} \cup I$}{
         $(T_1,\eta_1), \ldots, (T_h,\eta_h) \leftarrow ~\dec(T \setminus \{z\} \cup I,\eta)$\;
         \If{$\bigcup_{p \in [h]}\mathcal T[T_p,\eta_p] \neq$ None and $H[\eta_p(T_p),\eta_{p'}(T_{p'})]$ is $(\eta_p,\eta_{p'})$-isomorphic to $G[\mathcal T[T_p,\eta_p],\mathcal T[T_{p'},\eta_{p'}]]$, $\forall p \neq p' \in [h]$}{
           $\eta' \leftarrow x \in T \setminus \{z\} \mapsto \eta(x)$, $z \mapsto \eta(u) \cup \eta(v)$\;
           $\mathcal T[T,\eta'] \leftarrow \bigcup_{p \in [h]}\mathcal T[T_p,\eta_p]$\;
           \If{$\eta'(T) = [k]$}{
             \Return{$\mathcal T[T,\eta']$}\;
           }
         }
       }  
     }
   }
 }
 \Return{None}\;
 \caption{\algindsub, \algsubiso by changing \emph{isomorphic} to \emph{subisomorphic} (line 11)}
 \label{alg:indsub}
  \end{algorithm}

  \medskip

\textbf{Correctness.}
The soundness and completeness of \algindsub and \algsubiso follow as in the proof of~\cref{thm:k-mis}.
Therefore we only state the invariant maintained to show the completeness: After iteration $i$ (note that the first iteration is actually iteration $n-1$, and that the initialization is iteration $n$) of the outermost for loop, for every set $T \subseteq V(G_i)$ of size at most $|V(H)|=k$ connected in the red graph $(V(G_i),R(G_i))$, and every division $\eta$ of $T$, if there is a $(T,\eta)$-compliant set $S$, then $\mathcal T[T,\eta]$ contains such a set $S$.
In particular if we skip the possible exit of lines 14 and 15, after the last iteration (iteration 1), $\mathcal T[V(K_1),\eta: x \in V(K_1) \mapsto [k]]$ contains an actual set $S$ (and not ``None'') if and only if the \textsc{(Induced) Subgraph Isomorphism}-instance admits a solution.
The only ``new'' element (compared to \kmis) to prove the invariant is the potential presence of black edges between red connected components.
Nevertheless this was already evoked in the description of \algindsub and is dealt with straightforwardly.

\medskip

\textbf{Running time.}
There are four nested for loops in Algorithm~\ref{alg:indsub}.
The first one (outermost) brings a multiplicative $n$ factor to the overall running time, the second, an $d^{2k}$ factor (by~\cref{cor:connected-subgraphs}), the third one, a factor $3$.
The fourth and innermost for loop ranges over all the divisions of a fixed set $T'$ of size at most $k$.
($T'$ could in principle be of size $k+1$, but such sets can be automatically discarded since they do not admit any division.)
Every such division can be seen as a bijective mapping from $T'$ to the parts of a partition of a subset of $V(H)=[k]$.
There are at most $2^k B_k = 2^{O(k \log k)}$ partitions of a subset of $[k]$, where $B_k$ is the $k$-th Bell number.
Then there are at most $k^k=2^{k \log k}$ bijections from $T'$ to these parts.
Thus there are at most $2^{O(k \log k)}$ divisions, and the last for loop incurs a $2^{O(k \log k)}$ factor.

Decomposing $(T',\eta)$ and checking for a potential compliant solution can be done in time $k^{O(1)}$.
Thus the overall running time of \algindsub and \algsubiso is $2^{O(k \log k)}d^{2k}n = 2^{O_d(k \log k)}n$.
Again it can be observed that even when $d$ is polylogarithmic in $n$, this running time is FPT in $k$~\cite{Sloper08}.

As in \cref{thm:k-mis}, a better practical algorithm (with similar worst-case running time) consists of building the partial solutions in $G_i$ by unions of at most $\min(2,d+1)$ partial solutions in $G_{i+1}$ that are pairwise disconnected in the red graph and neighboring the vertices $u$ and $v$. 
\end{proof}

The \scaset problem on an input graph $G$ is equivalent to \kmis on $G^{\leqslant r}$.
The following theorem is a consequence that FO interpretations preserve bounded twin-width~\cite{twin-width1}.
As $G^{\leqslant r}$ can be obtained by FO interpretation $\phi$ of size $O(r)$ on $G$, $\tww(G^{\leqslant r}) \leqslant f(\tww(G),r)$.
Treating $d = \tww(G)$ and $r$ as constants, it is noteworthy that the complexity of \scaset remains the essentially optimal $2^{O(k)}n$.

\begin{theorem}\label{thm:k-scattered}
  Given a graph $G$, a $d$-sequence $G=G_n, G_{n-1},$ $\ldots, G_1=K_1$, \scaset can be solved in time $2^{O_{d,r}(k)} n$.
\end{theorem}

\section{A practical algorithm for \kds}\label{sec:kds}

We solve \kds with a more involved instantiation of the scheme of the previous section.
We face some new conceptual difficulties compared to the algorithm for \kmis.
For one thing, the partial solutions that we maintain are not feasible solutions in the whole graph.
Also we now consider balls of radius $f(d)k$ in the red graphs, and not merely of radius $k$.
In general, the arguments are more subtle to handle partially and fully dominated vertex sets, as well as the solution trace.
This entails a worse dependency in~$d$, but the same essentially optimal $2^{O(k)}n$ when $d$ is treated as a constant. 

\begin{theorem}\label{thm:kds}
  Given an $n$-vertex graph $G$, a positive integer $k$, and a $d$-sequence $G=G_n, \ldots, G_1=K_1$, \kds can be solved in time $O(2^{2(d^2+1)(2+\log d)k} n)=2^{O_d(k)} n$.
\end{theorem}

\begin{proof}
As was the case with \kmis, the algorithm sequentially considers each trigraph in the $d$-sequence $G_n, \ldots, G_1$ starting from $G_n$,   
and inductively updates a set of optimal partial solutions of the  trigraph $G_i$ to yield the next set  for $G_{i-1}$. 
We recall that $E(G_i)$ and $R(G_i)$ respectively refer to the black and red edge set of the trigraph $G_i$. 
The ball of radius at most $r$ in the red graph $(V(G_i),R(G_i))$ centered at a vertex $x\in V(G_i)$ is denoted as $B^r_i(x)$. 
 
 \medskip
 
\noindent {\bf Profile of a partial solution.} A \emph{profile (of a partial solution)} of $G_i$ is a triple $(T,D,M)$ of  vertex sets of $V(G_i)$ 
such that (i) $T$ forms a connected set in the red graph $(V(G_i),R(G_i))$, (ii) $D,M\subseteq T$, and (iii) $\bigcup_{x\in D}B^2_i(x)\subseteq T$. 
The first entry $T$ of a profile $P=(T,D,M)$ is called the \emph{ground set} of $P$, and the size of $P$ is defined as the size of its ground set. 
A profile $(T,D,M)$ is said to be a \emph{$k$-profile} if $|D|\leq k$. When the profile under consideration is clear from the context, 
we denote $T\setminus D$ and $T\setminus M$ by $\bar{D}$ and $\bar{M}$ respectively.

We say that \emph{a profile $(T,D,M)$ is realizable with $S\subseteq V(G)$} if the following conditions hold.
\begin{compactenum}
\item $S\subseteq \bigcup_{x \in T} x(G)$, 
\item for every $x \in V(G_i)$, $x \in D$ if and only if $x(G)\cap S\neq \emptyset$, and
\item for every $x \in V(G_i)$, $x \in M$ if and only if $x(G)$ is (fully) dominated by $S$.
\end{compactenum}
A profile is said to be \emph{realizable} if there exists $S$ with which it is realizable.

Suppose that $x,y \in V(G_{i+1})$ are contracted to yield $G_i$ with $z$ being the new vertex. 
For a vertex set $T\subseteq V(G_i)$ connected in the red graph $V(G_{i},R_{i})$ and containing $z$, 
let $T_1,\ldots , T_{\ell}$ be the red connected components of $T'=(T\setminus z)\cup \{x,y\}$ in $G_{i+1}$, i.e. 
the partition of $T'$ into maximal vertex sets each of which is connected in $V(G_{i+1},R_{i+1})$.
The number of these red subgraphs does not exceed $d+2$ because each $T_i$ either contains $x$ or $y$, or one of the newly created red neighbors of $z$. 
Notice also that $\ell$ can be equal to 1, which means that $x$ and $y$ belong to the same connected component of $(V(G_{i+1}),R(G_{i+1}))$.

For a $k$-profile $(T,D,M)$ of $G_i$ such that $z\in T$, we say that a set $\mathcal P=\{(T_1,D_1,M_1), \ldots,$ $(T_{\ell}, D_{\ell}, M_{\ell})\}$ of $k$-profiles of $G_{i+1}$ 
is \emph{consistent with $(T,D,M)$} if the following holds. Let $T':= (T\setminus z)\cup \{x,y\}$, $D':=\bigcup_{j=1}^{\ell} D_j$ and $M':=\bigcup_{j=1}^{\ell} M_j$.
\begin{compactenum}
\item The ground sets of the profiles in $\mathcal P$ are precisely the red components of $T'$ in $G_{i+1}$.
\item $D\setminus z=D'\setminus \{x,y\}$.
\item $z\in D$ if and only if $x\in D'$ or $y\in D'$. 
\item For every $u \in T\setminus z$, $u\in M$ if and only if $u \in M'$ or there exists $v \in D'$ such that $uv$ is a black edge in $G_{i+1}$.
\item $z \in M$ if and only if for each $u \in \{x,y\}$, it holds that: $u \in M'$ or there exists $v \in D'$ such that $uv$ is a black edge in $G_{i+1}$.
\end{compactenum}

\medskip

\noindent {\bf Algorithm, and how to compute $\tau_i$ from $\tau_{i+1}$.} At each iteration along the $d$-sequence, 
we maintain one mapping $\tau_i$ from $k$-profiles $P=(T,D,M)$ of $G_i$ with $|T|<(d^2+1)k$ to a subset of $\bigcup_{t\in T}t(G)$. 
The assignment $\tau_i(P) = nil$ is interpreted as that $P$ is not realizable whereas $\tau_i(P) \neq nil$ is intended to be a minimum-size vertex set of $V(G)$ realizing $P$. 
Again let $G_i$ be obtained by contracting the vertices $x,y\in V(G_{i+1})$ and $z$ be the new vertex. 
Our goal is to compute $\tau_i$ from $\tau_{i+1}$, assuming $\tau_{i+1}$ has been computed correctly. 
Note that a $k$-profile $P=(T,D,M)$ of $G_{i}$ such that $z\notin T$ is also a profile of $G_i$, and trivially one is realizable with $S$ if and only if the other is realizable with $S$. 
Therefore, $\tau_i$ simply inherits the assignment of $\tau_{i+1}$ in this case as depicted in lines 6-7.

If $P=(T,D,M)$ has $z$ in its ground set, the algorithm \algkds inspects all sets $\mathcal P$ of $k$-profiles of $G_{i+1}$ consistent with $(T,D,M)$
 and among the unions $\bigcup_{P\in \mathcal P} \tau_{i+1}(P)$ over all such $\mathcal P$, outputs the best one as $\tau_i(T,D,M)$, that is, the one of minimum cardinality is chosen.
If $\bigcup_{P\in \mathcal P} \tau_{i+1}(P)=nil$ for each consistent $\mathcal P$, the algorithm concludes that $(T,D,M)$ is not realizable and assigns $nil$.
The case when $\mathcal P$ contains a $k$-profile $P$ with ground set of size at least $(d^2+1)k$, a special step is taken as $\tau_{i+1}$ is not defined on such $P$. 
In this situation, a vertex $v\in T'\setminus \bigcup_{t\in D'}B^2_{i+1}(t)$ is chosen, and the query at $(T'\setminus v,D'\setminus v,M'\setminus v)$ is made instead. 
Lines 15-18 handle this case. The uniqueness of $k$-profile in $\mathcal P$ in line 16 and the existence of such $v$ in line 17 will be discussed in the correctness proof. 

\begin{algorithm}
  \DontPrintSemicolon
  \SetKwInOut{Input}{Input}
  \SetKwInOut{Output}{Output}
  \Input{~~A graph $G$, a positive integer $k$, and a $d$-sequence $G=G_n, \ldots, G_1=K_1$.}
  \Output{~~A dominating set of $G$ of size at most $k$, or report $nil$ (\textsc{No}-instance).}
 
 \For{$v\in V(G_n)$}{
  		$\tau_n(\{v\},\{v\},\{v\})=\{v\}$, $\tau_n(\{v\},\emptyset,\emptyset)=\emptyset$, $\tau_n(P) = nil$ for all other $k$-profiles $P$\;
  }
 \For{$i = n-1 \rightarrow 1$}{
   		$x,y \leftarrow $ contracted pair in $G_{i+1} \to G_i$\;
   		$z \leftarrow $ contraction of $x$ and $y$ in $G_i$\;
		\For{every  $k$-profile $(T,D,M)$ of $G_i$ of size less than $(d^2+1)k$ s.t. $z\notin T$}{
					$\tau_i(T,D,M) \leftarrow \tau_{i+1}(T,D,M)$\;
		}
 		 \For{every  $k$-profile $(T,D,M)$ of $G_i$ of size less than $(d^2+1)k$ s.t. $z\in T$}{
		 		$\tau_i(T,D,M) \leftarrow nil$\;
  				$T' \leftarrow (T\setminus z)\cup\{x,y\}$\;
				\For{every set $\mathcal P$ of $k$-profiles of $G_{i+1}$ consistent with $(T,D,M)$}{
						\If{each $k$-profile of $\mathcal P$ has size less than $(d^2+1)k$}{
								\If{$\tau_{i+1}(P)\neq nil$ for all $P\in \mathcal P$}{  
											$\tau_i(T,D,M) \leftarrow~\best\{\tau_i(T,D,M),  \bigcup_{P\in \mathcal P}\tau_{i+1}(P)\}$\;		
								}
  	 					}
						\Else{
								Let $(T',D',M')$ be the unique $k$-profile contained in $\mathcal P$.\; 
								Choose $v\in T'\setminus \bigcup_{t\in D'}B^2_{i+1}(t)$\;
								$\tau_i(T,D,M) \leftarrow~\best\{\tau_i(T,D,M), \tau_{i+1}(T'\setminus v,D'\setminus v,M'\setminus v)\}$ \;
						}
				}

				\If{$\tau_i(T,D,M)\neq nil$ and has size larger than $k$}{
							$\tau_i(T,D,M) \leftarrow nil$\;
				}
  		}
} 
 \Return{$\tau_1(V(G_1),V(G_1),V(G_1))$}\;
 \caption{\algkds}
 \label{alg:kds}
\end{algorithm}

\medskip

\noindent {\bf Correctness.} To show the correctness of Algorithm~\ref{alg:kds}, it suffices to prove the following.
\begin{quote}
$(\star)$ For every $i\in [n]$ and every $k$-profile $P$ of $G_i$, we have $\tau_i(P)\neq nil$ if and only if $P$ is realizable with a 
set of size at most $k$. Furthermore, if $\tau_i(P)\neq nil$, then $\tau_i(P)$ is a set of minimum size with which $P$ is realizable.
\end{quote} 
We prove $(\star)$ by induction. In the base case  when $i=n$, the claim trivially holds. Assume $i<n$ and let $x,y$ be the vertices of $G_{i+1}$ 
which were contracted to yield $G_i$, where $z$ is the newly obtained vertex of $G_i$. 
By induction hypothesis, for any $k$-profile $(T,D,M)$ of $G_i$ with  $z\notin T$ the claim holds as it is a $k$-profile of $G_{i+1}$ as well. 

Therefore, we assume that $z\in T$ and let $T'=(T\setminus z)\cup \{x,y\}$. 
\begin{claim}
Assume that $(\star)$ holds for all $i'>i$ and let $P=(T,D,M)$ be a $k$-profile of $G_i$. If $P$ is realizable with a set of size at most $k$, then $\tau_i(P)\neq nil$.
\end{claim}
\begin{proofofclaim}
Suppose that $P=(T,D,M)$ is realizable with $S\subseteq V(G)$ of size at most $k$. 
Let $T_1,\ldots , T_{\ell}$ be the  red connected components of $T'$ in $G_i$, and let $S_j= S \cap \bigcup_{t\in T_j}t(G)$ for every $j \in [\ell]$. 
The pairs $T_j$ and $S_j$ for $j=1, \ldots , \ell$ define a set of $\ell$ $k$-profiles $(T_j, D_j,M_j)$ of $G_{i+1}$ in a canonical way: 
$D_j$ is precisely the set of vertices $t\in T_j$ such that $t(G) \cap S_j$ 
and $M_j$ is the set of vertices $t\in T_j$ such that $t(G)$ is (fully) dominated by $S_j$. 
By construction, each $k$-profile $(T_j,D_j,M_j)$ is realizable with $S_j$. 

We argue that the set $\mathcal P=\{(T_j,D_j,M_j):j\in [\ell]\}$ is consistent with $P=(T,D,M)$.
The first and the second conditions for consistency are clearly satisfied.
To verify the third condition, consider a vertex $u \in T$ distinct from $z$ and without loss of generality we assume $u\in T_{j^*}$. 
If $u\in M$ and $u\notin M_{j^*}$, this means that $S_{j^*}$ does not dominate $u(G)$ because $S_{j^*}$ realizes $(T_{j^*},D_{j^*},M_{j^*})$. 
From $u\in M$ and the fact that $S$ realizes $(T,D,M)$, we know that $S$ dominates $u(G)$ and thus there is at least one vertex $S\setminus S_{j^*}$ which is adjacent (in $G$) with some vertex of $u(G)$.
Consider an arbitrary vertex $v\in T$ to which some of $S\setminus S_{j^*}$ contracts to, and observe that $v \notin T_{j^*}$.
This means that $uv$ is a black edge.
The converse direction of the third condition is clearly met.
The fourth condition of consistency can be verified similarly as the third condition.

If $\mathcal P$ does not contain any $k$-profile whose ground set has size at least $(d^2+1)k$, now the claim is immediate 
because each $(T_j,D_j,M_j)$ is realizable with $S_j$: by induction hypothesis, we have $\tau_{i+1}(T_j,D_j,M_j)\neq nil$, and thus $\tau_i(T,D,M)$ is set to $\neq nil$ at line 14.

Suppose that $\mathcal P$ contains a $k$-profile whose ground set has size at least $(d^2+1)k$. One can easily see that 
in this case, $\ell=1$ or equivalently $T'$ is a red connected component  in $(V(G_{i+1}),R(G_{i+1}))$ consisting of exactly $(d^2+1)k$ vertices. Since 
the union of at most $k$ balls of radius at most $2$ which is connected in  $(V(G_{i+1}),R(G_{i+1}))$ have less than $(d^2+1)k$ vertices, 
there exists $v\in T'\setminus \bigcup_{t\in D'}B^2_{i+1}(t)$. Moreover, by the choice of $v$, $(T'\setminus v,D'\setminus v,M'\setminus v)$ is now a $k$-profile of $G_{i+1}$. 
To conclude that $\tau_i(T,D,M)\neq nil$, it suffices to prove that $\tau_{i+1}(T'\setminus v,D'\setminus v,M'\setminus v)\neq nil$. 
We do this by showing that $(T,D,M)$, $(T',D',M')$ and $(T'\setminus v,D'\setminus v,M'\setminus v)$ are equivalent in regards to realizability. 

The equivalence of the first two is obvious. 
For the equivalence of the last two, note that if $S$ realizes $(T',D',M')$,  
$S$ does not intersect $v(G)$, and thus $S$ trivially realizes $(T'\setminus v,D'\setminus v,M'\setminus v)$. 
Conversely, suppose that $(T'\setminus v,D'\setminus v,M'\setminus v)$ is realizable with $S'$.
The crucial observation is that $v$ has no red neighbor in $D'$ since otherwise, 
$v$ belongs to the union $\bigcup_{t\in D'}B^2_{i+1}(t)$, contradicting the choice of $v$. 
Therefore, we know that $v\in M'$ if and only if there exists $u\in D'\setminus v$ such that $uv$ is a black edge. In the case when $v\in M'$, there exists 
a black neighbor $u\in D'\setminus v$ of  $v$, and any $S'$ 
realizing $(T'\setminus v,D'\setminus v,M'\setminus v)$ intersects $u(G)$. If follows that $S'$ fully dominates $v(G)$ and 
$S'$ realizes $(T',D',M')$.
Else if $v\notin M'$, this means that not only the red neighbors of $v$ are disjoint from $D'$ but also 
no black neighbor of $v$ is contained in $D'$. As a consequence $v(G)$ is not dominated by $S'$, thus $S'$ realizes $(T',D',M')$. 
This proves the equivalence of $(T',D',M')$ and $(T'\setminus v,D'\setminus v,M'\setminus v)$, and completes the proof of the claim.
\end{proofofclaim}

To establish the other direction, suppose that $\tau_i(T,D,M)\neq nil$ and let $\mathcal{P}^*$ be the set consistent with $P$ such that 
$\tau_i(T,D,M)=\bigcup_{P\in \mathcal{P}^*} \tau_{i+1}(P)$ or $\tau_i(T,D,M)=\tau_{i+1}(T'\setminus v,D'\setminus v,M'\setminus v)$ for some $v$. 
Such $\mathcal{P}^*$ clearly exists since otherwise only $nil$ can be output. In the former case, it is  tedious to verify that 
if each $(T_i,D_i,M_i)$ of $\mathcal{P}^*$ is realizable with $S_i$, then $\bigcup_{i \in [\ell]} S_i$ realizes $(T,D,M)$. 

In the latter case, we simply recall that $(T,D,M)$ and $(T'\setminus v,D'\setminus v,M'\setminus v)$ are equivalent in regards to realizability. 
This completes the proof of the first statement of $(\star)$.
The second statement immediately follows.

\medskip

\noindent {\bf Running time.}
In an actual implementation of Algorithm~\ref{alg:kds}, we maintain a single mapping~$\tau$.
As we proceed from $G_{i+1}$ to $G_i$, we modify the domain of $\tau$ consisting of $k$-profiles so that new $k$-profiles involving $z$ are added and after calculating the assignments for the new $k$-profiles, all the domains and corresponding assignments involving $x$ or $y$ shall be discarded. 
Therefore, it suffices to check the running time for updating $\tau$, which is performed in the inner loop of lines 6-20. 
By~\cref{cor:connected-subgraphs}, there are $O(d^{2(d^2+1)k-2}\cdot 2^{2(d^2+1)k})$ new profiles of $G_i$ to compute. 
For each $k$-profile $(T,D,M)$ with $z\in T$, the ground sets $T_1,\dots, T_{\ell}$ of a potentially consistent set $\mathcal P$ is already determined. 
Hence, we exhaust all possibilities of appending each $T_i$ by $M_i$ and $D_i$ to form a $k$-profile and the inner loop of 8-20 will consider at most 
$2^{(d^2+1)k}\cdot 2^{(d^2+1)k}$ sets $\mathcal P$. The consistency of $\mathcal P$ with $(T,D,M)$ can be routinely verified.
This establishes the claimed running time.
\end{proof}

\section{Bounded twin-width classes are $\chi$-bounded}\label{sec:chibounded}

So far, our algorithms followed the same recipe: Initialize partial solutions on single-vertex sets, stitch together a bounded number of partial solutions when they become connected in the red graph after the current contraction, and conclude with the partial solutions on the last (1-vertex) graph of the sequence.
This is the original scheme of Guillemot and Marx~\cite{Guillemot14}, and of our model checking algorithm~\cite{twin-width1}.

We now present a novel use of the contraction sequence.
It consists of starting at the end, when all the vertices are contracted on a single vertex, and rewinding the sequence.
The single vertex is first ``split'' into two vertices (linked by a black or red edge if $G$ is connected).
Then one of these two vertices is split into two new vertices, and so on.
Typically, at first, edges are mostly red.
As the vertex partition gets finer, black edges start appearing (eventually all edges are black).
In this direction of time, black edges are irreversible: When a black edge first appears between $x$ and $y$ in $V(G_i)$, it stays or rather spreads into the biclique $(x(G),y(G))$.
We use this new viewpoint to color triangle-free graphs of bounded twin-width with a constant number of colors.
We show that the newly split vertices can be greedily colored, while the rest of the colors remains unchanged.
Importantly for coloring, in a triangle-free graph, when a black edge appears between $x$ and $y$ we know that both sides $x(G)$ and $y(G)$ of the biclique are independent sets. 

The following coloring procedure essentially contains the $\chi$-boundedness of bounded twin-width classes.
Despite its simplicity, this for instance generalizes the non-trivial result that bounded rank-width classes are $\chi$-bounded~\cite{Dvorak12}.
The proof that graphs with bounded rank-width have bounded twin-width, presented in \cite{twin-width1}, is also elementary. 

\begin{theorem}\label{thm:triangle-free}
Every triangle-free graph with twin-width at most~$d$ is $d+2$-colorable.  
\end{theorem}
\begin{proof}
  Let $G$ be an $n$-vertex triangle-free graph of twin-width at most~$d$, and let $G=G_n,\ldots,G_1=K_1$ be a $d$-sequence of $G$.
  We show how to color $G$ with $d+2$ colors starting from $G_1$, and iteratively coloring $G_{i+1}$ based on the coloring of $G_i$.
  We give the unique vertex of $G_1=K_1$ color 1.
  This defines coloring $C_1$.
  For every $i$ from 1 to $n-1$, let $z$ be the vertex of $G_i$ split into $u, v \in V(G_{i+1})$.
  In coloring $C_{i+1}$, every vertex of $V(G_{i+1}) \setminus \{u,v\}$ keeps the color it received by $C_i$.
  Vertex $u$ receives color $C_i(z)$.
  Finally, $v$ receives color $C_i(z)$ if $uv$ is a non-edge in $G_{i+1}$, and the smallest positive integer \emph{not} appearing in its neighborhood (black \emph{and} red neighbors) in $G_{i+1}$, otherwise. 
  We will now show that $C_n$ is a proper coloring of $G$ using at most $d+2$ distinct colors.

  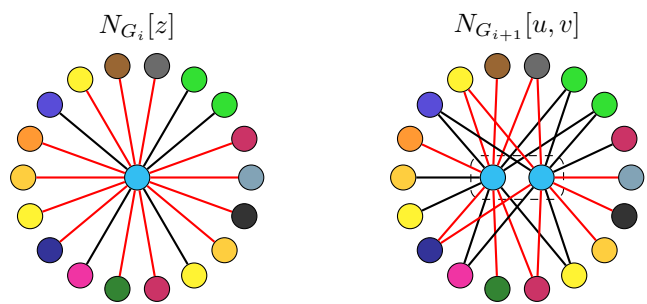
\begin{figure}
    \centering
    \begin{tikzpicture}
      \def\s{1.5} ; 
      \foreach \a/\c [count = \i from 1] in {
        0/cyan!60!purple,20/purple,40/black!15!green,60/black!15!green,80/gray!55!black,100/orange!50!black,120/yellow,140/blue!75!brown,160/orange,180/yellow!50!orange,
        200/yellow,220/blue!50!black,240/magenta,260/black!60!green,280/purple,300/yellow,320/yellow!50!orange,340/black}
               {
        \node[draw,circle,fill opacity=0.8,fill=\c] (v\i) at (\a:\s) {} ;
      }
      \node[draw,circle,fill opacity=0.8,fill=cyan] (z) at (0,0) {} ;

      \foreach \i in {3,4,8,13,16}{
        \draw[thick] (z) -- (v\i) ;
      }
      \foreach \i in {1,2,5,6,7,9,10,11,12,14,15,17,18}{
        \draw[thick,red] (z) -- (v\i) ;
      }
      \node at (0,2) {$N_{G_i}[z]$} ; 

      \begin{scope}[xshift=5cm]
        \foreach \a/\c [count = \i from 1] in {
           0/cyan!60!purple,20/purple,40/black!15!green,60/black!15!green,80/gray!55!black,100/orange!50!black,120/yellow,140/blue!75!brown,160/orange,180/yellow!50!orange,
           200/yellow,220/blue!50!black,240/magenta,260/black!60!green,280/purple,300/yellow,320/yellow!50!orange,340/black}
                 {
        \node[draw,circle,fill opacity=0.8,fill=\c] (v\i) at (\a:\s) {} ;
      }
            \node[draw,circle,fill opacity=0.8,fill=cyan] (u) at (-0.32,0) {} ;
            \node[draw,circle,fill opacity=0.8,fill=cyan] (v) at (0.32,0) {} ;
       \node[draw,dashed,rounded corners,fit=(u) (v)] {} ;
            
      \foreach \i in {3,4,8,10,11,13,16}{
        \draw[thick] (u) -- (v\i) ;
      }
      \foreach \i in {5,6,7,9,12,14,15}{
        \draw[thick,red] (u) -- (v\i) ;
      }
       \foreach \i in {2,3,4,8,13,16}{
        \draw[thick] (v) -- (v\i) ;
      }
      \foreach \i in {1,5,7,12,15,17,18}{
        \draw[thick,red] (v) -- (v\i) ;
      }
      \node at (0,2) {$N_{G_{i+1}}[u,v]$} ; 
      \end{scope}
    \end{tikzpicture}
    \caption{Split, when $z$ is incident to a black edge in $G_i$. As $G$ is triangle-free, there cannot be an edge (red or black) between $u$ and $v$. Thus both $u$ and $v$ can take the color of $z$, which does not appear in their neighborhood.}
    \label{fig:coloring1}
  \end{figure}
   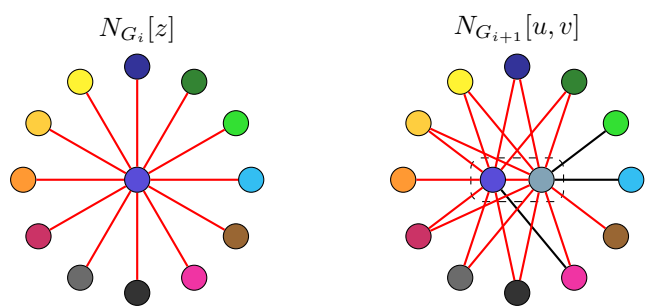
\begin{figure}
    \centering
    \begin{tikzpicture}
      \def\s{1.5} ; 
      \foreach \a/\c [count = \i from 1] in {
        0/cyan,30/black!15!green,60/black!60!green,90/blue!50!black,120/yellow,150/yellow!50!orange,180/orange,
        210/purple,240/gray!55!black,270/black,300/magenta,330/orange!50!black}{
        \node[draw,circle,fill opacity=0.8,fill=\c] (v\i) at (\a:\s) {} ;
      }
      \node[draw,circle,fill opacity=0.8,fill=blue!75!brown] (z) at (0,0) {} ;

      \foreach \i in {1,...,12}{
        \draw[red,thick] (z) -- (v\i) ;
      }
      \node at (0,2) {$N_{G_i}[z]$} ; 

      \begin{scope}[xshift=5cm]
        \foreach \a/\c [count = \i from 1] in {
        0/cyan,30/black!15!green,60/black!60!green,90/blue!50!black,120/yellow,150/yellow!50!orange,180/orange,
        210/purple,240/gray!55!black,270/black,300/magenta,330/orange!50!black}{
        \node[draw,circle,fill opacity=0.8,fill=\c] (v\i) at (\a:\s) {} ;
      }
        \node[draw,circle,fill opacity=0.8,fill=blue!75!brown] (u) at (-0.32,0) {} ;
        \node[draw,circle,fill opacity=0.8,fill=cyan!60!purple] (v) at (0.32,0) {} ;
        \draw[red,thick] (u) -- (v) ;
        \node[draw,dashed,rounded corners,fit=(u) (v)] {} ;

      \foreach \i in {3,...,10}{
        \draw[red,thick] (u) -- (v\i) ;
      }
      \foreach \i in {11}{
        \draw[thick] (u) -- (v\i) ;
      }
      \foreach \i in {3,...,6,8,9,...,12}{
        \draw[red,thick] (v) -- (v\i) ;
      }
       \foreach \i in {1,2}{
        \draw[thick] (v) -- (v\i) ;
      }
      \node at (0,2) {$N_{G_{i+1}}[u,v]$} ;
      \end{scope}
    \end{tikzpicture}
    \caption{Split, when $z$ is only incident to red edges. Even if the red neighbors of $z$ have $d$ distinct colors, vertex $v$ can find a color in $[d+2]$ which avoids these $d$ colors plus the color of $z$ and $u$.}
    \label{fig:coloring2}
  \end{figure}
  
  We show by induction on $i$ that $C_i$ is a proper $d+2$-coloring of the \emph{graph} $G'_i := (V(G_i),E(G_i) \cup R(G_i))$. 
  Coloring $C_1$ is indeed proper in $G'_1$ and uses $1 \leqslant d+2$ color.
  We assume that $C_i$ is a proper $d+2$-coloring of $G'_i$, and distinguish two cases.
  If there is a black edge $yz \in E(G_i)$ (recall that $z$ is the vertex split into $u, v$), then $uv$ has to be a non-edge in $G_{i+1}$.
  Otherwise there is at least one edge between $u(G)$ and $v(G)$, and this edge forms a triangle with any vertex in $y(G)$.
  Thus in that case, $C_{i+1}(u) = C_{i+1}(v) = C_i(z)$.
  So the number of distinct colors given by $C_{i+1}$ is still at most $d+2$ (see~\cref{fig:coloring1}).
  And $C_{i+1}$ is a proper coloring of $G'_{i+1}$ since $N_{G'_{i+1}}(\{u,v\}) = N_{G'_i}(z)$. 
  If instead $z$ has only red neighbors in $G_i$, then $z$ has at most $d$ neighbors in $G'_i$.
  Furthermore let us assume that $uv \in E(G'_{i+1})$, otherwise we conclude as previously.
  In that case, $v$ is properly colored by $C_{i+1}$ in $G'_{i+1}$ by construction, and vertex $u$ as well, since $N_{G'_{i+1}}(u) \setminus \{v\} \subseteq N_{G'_i}(z)$.
  Finally $C_{i+1}(v)$ is the smallest positive integer not appearing in a set of at most $d+1$ positive integers.
  Thus $C_{i+1}(v) \leqslant d+2$, and $C_{i+1}$ is overall a proper $d+2$-coloring of $G'_{i+1}$ (see~\cref{fig:coloring2}).

  In particular, $C_n$ is a proper $d+2$-coloring of $G'_n = G_n = G$.
\end{proof}

As a side note, it is, to our knowledge, possible that every triangle-free $K_t$-minor free graph has twin-width $O(t)$.
If this turns out to be true, it offers a seemingly different approach to getting improved bounds in the triangle-free case of the Hadwiger's conjecture:
Instead of trying to color these graphs, one could try to design contraction sequences for them.

We now show how to color any $K_t$-free graph $G$ given with a $d$-sequence, with at most $(d+2)^{t-2}$ colors.
We use the scheme of \cref{thm:triangle-free} and color some induced subgraphs of $G$ by induction on $t$.

\begin{theorem}\label{thm:chibounded}
 For every integer $t \geqslant 3$, every $K_t$-free graph with twin-width at most~$d$ is $(d+2)^{t-2}$-colorable.
\end{theorem}
\begin{proof}
  Let $G_n, \ldots, G_1$ be a $d$-sequence of a $K_t$-free graph $G$ with $t \geqslant 3$.
  In~\cref{thm:triangle-free}, whenever a vertex $x \in V(G_{i+1})$ was incident to a black edge for the first time (going from $G_1$ to $G_n$), the color of all the vertices in $x(G)$ was eventually set to the same value, namely $C_{i+1}(x)$.
  Now such a set $x(G)$ is not necessarily an independent set, but rather induces a $K_{t-1}$-free graph.
  Indeed, a $K_{t-1}$ in $G[x(G)]$ would form a $K_t$ in $G$ with any vertex of $y(G)$, where $xy \in E(G_{i+1})$.
  By induction on $t$, we may color $G[x(G)]$ with tuples of at most $t-3$ integers of $[d+2]$, and prepends $C_{i+1}(x)$ to these tuples.
  The base case $t=3$ is~\cref{thm:triangle-free}.
  We make the general idea a bit more precise.

  For every $i \in [n]$, we define $G^*_i$ as the \emph{graph} obtained from $G_i$ by blowing every vertex $x \in V(G_i)$ into $G[x(G)]$ whenever $x$ is incident to a black edge, and then turning every red edge into a black edge.
  We define the successive colorings $C'_1, \ldots, C'_n$ of $G^*_1, \ldots, G^*_n$, respectively, following the algorithm of~\cref{thm:triangle-free}.
  While there are no black edge in the current trigraph $G_i$, we set $C'_i := C_i$, where $C_i$ is the coloring in the triangle-free case. 
  Say, at least one black edge appears for the first time in $G_{i+1}$ (this is well-defined since $G_n$ has only black edges).
  Again we adopt the convention that $z \in V(G_i)$ was split into $u, v \in V(G_{i+1})$.
  Let $S$ be the set of (at most $d+2$) vertices with an incident black edge in $G_{i+1}$.
  (One may notice that $S \subseteq \{u,v\} \cup N_{G_i}(z)$ and $S \cap \{u,v\} \neq \emptyset$.)
  Every vertex $w \in V(G_{i+1}) \setminus S$ receives color $C_{i+1}(w)$.  
  As we observed, for every $x \in S$, $G[x(G)]$ is $K_{t-1}$-free.
  By induction there is a coloring $C^x$ of $G[x(G)]$ with tuples of at most $t-3$ integers from $[d+2]$.
  We \emph{permanently} color every vertex $y \in x(G)$ by $(C_{i+1}(x),C^x(y))$.
  This defines the coloring $C'_{i+1}$ of $G^*_{i+1}$. 

  We continue to follow~\cref{thm:triangle-free}, with the ensuing precisions.
  We go through all the splits, including the ones between two permanently colored vertices, since they may make some other vertices incident to a black edge for the first time.
  If the split vertex $z \in V(G_j)$ is \emph{not} such that $z(G)$ was already permanently colored, the colors of the new vertices $u, v \in V(G_{j+1})$ are chosen according to the rules of~\cref{thm:triangle-free} where we consider the trigraphs $G_j$ and $G_{j+1}$ (and not the graphs $G^*_j$ and $G^*_{j+1}$), and the coloring $C_j$ of $V(G_j)$ is defined as: $C_j(y)$ is the first coordinate of $C'_j(y)$ (or $C'_j(y)$ itself if it is not a tuple) if $y \in V(G^*_j)$, and the first coordinate of the color of any vertex in $y(G)$, otherwise. 
  (One may observe that $C_j$ is \emph{not} necessarily a proper coloring of $(V(G_j),E(G_j) \cup R(G_j))$, but all the conflict edges lie within a permanently colored subgraph.)
  Every time a vertex $x$ becomes incident to a black edge, we permanently color $x(G)$.
  This defines the sequence of colorings $C'_1, \ldots, C'_n$.

  We show by induction on $i$ that $C'_i$ properly colors $G^*_i$.
  Coloring $C'_1$ is indeed a proper coloring of $G^*_1=K_1$.
  We assume that $C'_i$ is a proper coloring of $G^*_i$, and let $xy$ be any edge in $E(G^*_{i+1})$.
  By the outermost induction on $t$, if $xy$ lies within a $K_{t-1}$-free graph permanently colored, then $C'_{i+1}(x) \neq C'_{i+1}(y)$.
  If instead $x$ and $y$ belong to two distinct vertices of $G_{i+1}$, by the proof of~\cref{thm:triangle-free} and the fact that $C'_i$ is a proper coloring of $G^*_i$, the first coordinate of $C'_{i+1}(x)$ and of $C'_{i+1}(y)$ differ.
  
  In particular $C'_n$ is a proper coloring of $G^*_n = G_n = G$.
  We pad every tuple $C'_n(x)$ of length $t' < t$ with $t-t'$ entries 1.
  From the previous proof, it can be observed that this new coloring of $G$ is still proper, and uses at most $(d+2)^{t-2}$ colors.
\end{proof}

\cref{thm:chibounded} directly implies that, provided $O(1)$-sequences are given, \textsc{Min Coloring} can be $2^{O(\text{OPT})}$-approximated on bounded twin-width graphs, and \lmis can be $O(1)$-approximated on $K_t$-free graphs of bounded twin-width (trivially because an independent set of size $n/O(1)$ can be found).
In~\cref{subsec:mis-barrier,subsec:conclusion-app} we discuss further the approximability of \mis in bounded twin-width graphs.

It would be interesting to determine if bounded twin-width classes are polynomially $\chi$-bounded, that is, satisfies for some constant $c$, $\chi(G) = O(\omega(G)^c)$ for every graph $G$ in the class.
Bounded clique-width or rank-width classes were shown polynomially $\chi$-bounded only recently~\cite{Bonamy19}.
We show however that bounded twin-width classes satisfy the related \emph{strong Erd\H{o}s-Hajnal property}.
We recall that a class $\mathcal C$ of graphs satisfies the \emph{strong Erd\H{o}s-Hajnal property} if there exists an $\varepsilon > 0$ such that every $G \in \mathcal C$ contains two disjoint subsets of vertices $X, Y$, both of size at least $\varepsilon |V(G)|$, with either all edges or no edges between $X$ and~$Y$.
The strong Erd\H{o}s-Hajnal property of a hereditary class implies the existence of a clique or a stable set of polynomial size, that is, the Erd\H{o}s-Hajnal property~\cite{ALON2005310}.

\begin{theorem}\label{theo:EH}
The class of graphs with twin-width at most $d$ satisfies the strong Erd\H{o}s-Hajnal property with $\varepsilon = 1/(d+4)$.
\end{theorem}

\begin{proof}
  Let $G$ be an $n$-vertex graph with twin-width at most~$d$.
  Consider in a fixed $d$-sequence $G_n, \ldots, G_1$ the maximum index $i$ such that there is a vertex $z \in V(G_i)$ satisfying $|z(G)| \geqslant n/(d+4)$.
  Since $X := z(G)$ is the union of $u(G)$ and $v(G)$ for some $u, v \in V(G_{i+1})$, its size is at most $2n/(d+4)$.
  Vertex $z$ has at most $d$ red neighbors in $G_i$.
  These neighbors constitute a set $S \subseteq V(G)$ of at most $d \cdot n/(d+4)$ vertices.
  Thus $|V(G) \setminus (z(G) \cup S)| \geqslant n - 2n/(d+4) - dn/(d+4) = 2n/(d+4)$.
  By construction, every vertex in $V(G) \setminus (z(G) \cup S)$ is fully adjacent to $X$ or fully non-adjacent to $X$.
  Let $Y \subseteq V(G) \setminus (z(G) \cup S)$ be the subset of all vertices in the majority regarding these two outcomes.
  Set $Y$ has size at least $n/(d+4)$ vertices and $X, Y$ is therefore an appropriate pair.
\end{proof}

\section{Interval biclique partitions and computing shortest paths}\label{sec:ibp-sp}

In this section, we show how to build on the viewpoint of the previous section to compute shortest paths efficiently.
We first show that bounded twin-width graphs admit favorable edge partitions into linearly many bicliques.

An \emph{\bip} (or \emph{\sbip}, for short) of a graph $G$ on vertex set $[n]$ is a set $\mathcal B$ of bicliques that edge-partitions $G$ where each biclique $(A_i,B_i) \in \mathcal B$ is such that both \emph{sides} $A_i$ and $B_i$ are two (disjoint) discrete intervals of $[n]$ (see~\cref{fig:ibp}).
Observe that the latter condition makes \bips a more restricted form of the mere biclique (edge-)partitions.
However every graph admits an \sbip, since a biclique of $\mathcal B$ can be a single edge of $G$.
Such an edge-partition becomes interesting when the number of bicliques in $\mathcal B$ is small, say, at most linear in the number of vertices.
We will show that bounded twin-width graphs admit linear-sized \sbips.
To give an example, the clique $K_n$ admits $\{([1],[2,n]),([2],[3,n]),([3],[4,n]),\ldots,([n-1],[n])\}$ as an \sbip.
The \sbip $\mathcal B$ gives a $4 \lceil \log n \rceil |\mathcal B|$-bits representation of the graph.

The \emph{ordered union tree} of a $d$-sequence $\mathcal S: G=G_n, \ldots, G_1=K_1$, is a pair $(\mathcal T, \mathcal A)$ where $\mathcal T$ is a rooted binary tree whose leaves are in one-to-one correspondence with $V(G)$, and $\mathcal A$ is an array of length $n-1$ whose $i$-th entry is a pointer to the (distinct) internal node of $\mathcal T$ \emph{representing the $i$-th contraction of $\mathcal S$}, i.e., whose rooted subtree has for leaves all the vertices of $G$ ``contained'' in the contracted vertex.
Our algorithms in~\cite{twin-width1,twin-width2} can output an ordered union tree in the same running time as for computing the $d$-sequence.
The ordered union tree can thus be seen as an alternative way of presenting the $d$-sequence. 

\begin{lemma}\label{lem:seq-to-bir}
  Every $n$-vertex graph of twin-width~$d$ has an \bip $\mathcal B$ of size at most~$(d+1)(n-1)$.
  Furthermore $\mathcal B$ can be computed in time~$O(d n)=O_d(n)$ given the ordered union tree of a~$d$-sequence for~$G$. 
\end{lemma}
\begin{proof}
  We relabel the nodes of the tree starting from the leaves.
  From left to right, their label now describes the integers from 1 to $n$ (see~\cref{fig:union-tree}).
  An internal node gets label $[i,j]$ if the leaves of its subtree precisely form the interval $[i,j]$. 
  This step can be done in $O(n)$-time.

  Now we read the $d$-sequence backwards, starting from the end $G_1=K_1$, and tracking black edges appearing for the first time.
  Let $u, v \in V(G_{i+1})$ be obtained by splitting $z \in V(G_i)$.
  Formally we say that a \emph{black edge $xy \in E(G_{i+1})$ appears for the first time} in $G_{i+1}$, if $xy$ is not a black edge of $G_i$ (this implies that $\{x,y\} \cap \{u,v\} \neq \emptyset$) and $xy$ is not of the form $uy$ or $vy$ with $zy \in E(G_i)$.
  Intuitively, not only the black edge is new, but it did not originate from a black edge $zy \in E(G_i)$.
  Note that the latter automatically creates two black edges $uy, vy \in E(G_{i+1})$, but the information carried by these edges is contained in the biclique $(y(G),z(G))$ already detected. 
  
  At each of the $n-1$ steps, at most $d+1$ black edges can appear for the first time: possibly one between the two vertices $u, v$, and at most one between $\{u,v\}$ and every red neighbor of $z$ in $G_i$.
  We append the corresponding bicliques to $\mathcal B$.
  This takes overall time $O(d n)$, and shows that $|\mathcal B| \leqslant (d+1)(n-1)$.
  By the previous relabeling, the two sides of the bicliques are discrete intervals.
  By the final observation in the previous paragraph, the bicliques of $\mathcal B$ cover all the edges of $G$.
  By the definition of a ``black edge appearing for the first time'', no edge is covered twice, so $\mathcal B$ is indeed a biclique partition of $E(G)$. 
\end{proof}

\begin{figure}
  \centering
  \begin{tikzpicture}
    \def\s{1}
    \def\r{0.16}
     \foreach \i/\j/\k [count = \c from 1] in {4/1/23,4/2/21,4/3/15,4/4/14,5/1/24,5/2/22,5/3/16,5/4/19, 0/1/8,0/2.5/6,0/4/2,-1/1/9,-1/2/5,-1/3/4,-1/4/1, 1/0/26,2/0/27,3/0/30,0/0/25,1/-1/29,2/-1/28,3/-1/32,4/0/31, 1/2/7,2/1/13,2/2/12,2/3/11,3/2/20,3/3/18,1/4/3,2/4/10,3/4/17}{
      \node[draw,circle,minimum width=4.3mm,inner sep=0pt] (\c) at (\s * \i,\s * \j) {\k} ;
     }
     \foreach \i/\j in {{19,23,22,20}/0.4,{32,8,5,1,28}/0.4,{15,30,24,9,12}/0.4,{31,27}/0.4,{26,25}/0.4,
       {2,6,5,1}/0.32,{32,8,7,29}/0.32,{14,24,9,13}/0.32,{15,30}/0.32,{16,17,21,20}/0.32,{18,23}/0.32,
       {3,4,7}/0.24,{1,6}/0.24,{10,24,9}/0.24,{17,21}/0.24}{
        \draw \convexpath{\i}{\j cm} ;
     }
     \foreach \i/\j in {{2,2.4}/{2,2.6},{2,0.4}/{3.18,1.28},{0.82,1.5}/{1.6,1.5},{2.4,1.5}/{3.68,1.5},{0.78,0.22}/{0.22,0.78},4/7,18/23,18/22,11/30,32/29,16/20,25/26,9/24,{-0.78,2}/{-0.24,2},{-0.85,1.15}/{-0.73,1.27}, {0,2.85}/{0,3.68},{3,2.22}/{3,2.68}, {4.5,2.32}/{4.5,2.76}, {4.15,1.85}/{4.33,1.67}, {4.85,3.85}/{4.67,3.67}, {2.23,-0.5}/{2.75,-0.2}, {1.405,3.5}/{1.595,3.5}}{
       \draw[very thick,blue] (\i) -- (\j) ;
     }
  \end{tikzpicture}
  \caption{Example of an \bip following a contraction sequence.
    The bicliques are represented in bold blue.
    See~\cref{fig:union-tree} for a part of the corresponding union tree.}
  \label{fig:ibp}
\end{figure}
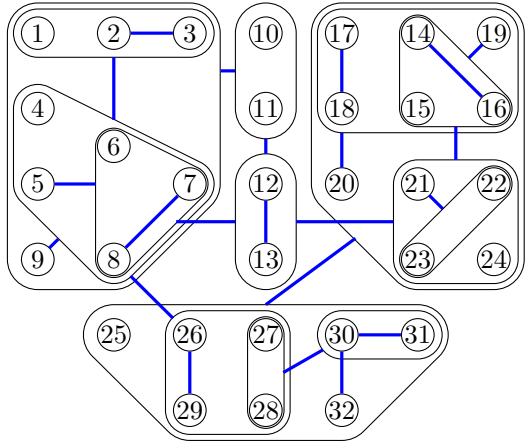
\begin{figure}
  \centering
  \begin{tikzpicture}
    \def\s{1}
    \foreach \i/\j/\k [count = \c] in {0/4/{[1,9]},-2/3/{[1,3]},2/3/{[4,9]},-3/2/{[1]},-1/2/{[2,3]}, -1.5/1/{[2]}, -0.5/1/{[3]}, 1/2/{[4,8]},3/2/{[9]}, 0.5/1/{[4,5]}, 2.5/1/{[6,8]}, 0/0/{[4]}, 1/0/{[5]}, 2/0/{[6]}, 3/0/{[7,8]}, 2.5/-1/{[7]}, 3.5/-1/{[8]}}{
      \node (\c) at (\i * \s, \j * \s) {$\k$} ;
    }
    \foreach \i/\j in {1/2,1/3,2/4,2/5,5/6,5/7,3/8,3/9,8/10,8/11,10/12,10/13,11/14,11/15,15/16,15/17}{
      \draw[thin] (\i) -- (\j) ;
    }
    \foreach \i/\j in {2/8,8/9,6/7,16/17,11/13}{
      \draw[very thick,blue] (\i) -- (\j) ;
    }
  \end{tikzpicture}
  \caption{The subtree $[1,9]$ of the union tree corresponding to the graph and sequence of \cref{fig:ibp}.
    The bicliques of the \sbip are represented in bold blue.
    The order of the splits does not appear.}
  \label{fig:union-tree}
\end{figure}
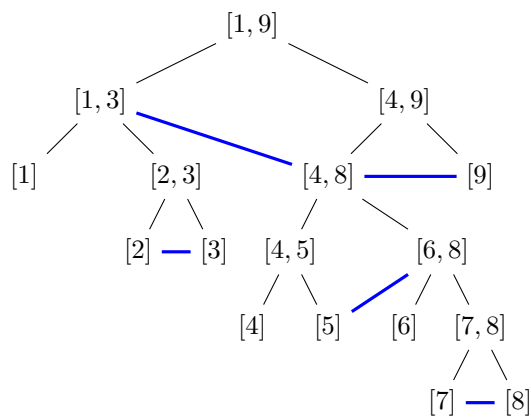

An interesting additional property of the computed \sbip $\mathcal B$, in the case of bounded twin-width graphs, is that the whole set of biclique sides (partite sets) defines a laminar family.
Indeed, by definition of a contraction sequence, there cannot be two overlapping sides.
Our algorithm will not use this additional property.

For the next algorithm, the \bip $\mathcal B$ is stored in a look-up table $\mathcal Z$.
One accesses in constant time, with $\mathcal Z[A]$, the head of the list of sides $B$ such that $(A,B)$ or $(B,A)$ is in $\mathcal B$.
The table $\mathcal Z$ can be initialized in time $O(|\mathcal B|)$, given the list of bicliques $\mathcal B$.

\begin{theorem}\label{thm:bir-to-sp}
  Given an \sbip $\mathcal B$ of an $n$-vertex graph $G$ and a vertex $s \in V(G)$, \sssp can be solved in time $O((n+|\mathcal B|)\log n)$.
\end{theorem}
\begin{proof}
  Essentially we will perform a breadth-first search (BFS) from vertex $s$, following the bicliques instead of single edges.
  
  To start with, we need a quick access to all the bicliques of $\mathcal B$ containing a given vertex $u \in V(G)$.
  As the sides of the bicliques are intervals, we in fact want to solve the interval stabbing problem: Preprocess a set $\mathcal I$ of intervals to answer queries of the form ``list all intervals of $\mathcal I$ containing $p$''.
  For instance, if we query vertex 21 of~\cref{fig:ibp}, we want the fast output of the list of intervals $[14,24],[21,24],[21]$.
  That way we can then get the neighborhood of 21 in the compact form $[12,13],[14,16],[22,23],[25,32]$. 
  Since our intervals range over~$[n]$, there are optimal \emph{static} data structures for that problem, with preprocessing time $O(n)$ and query time $O(q)$ where $q$ is the number of output intervals and $n$ is the total number of intervals (see for instance~\cite{Schmidt09} and \cite{Chazelle86}).
  To our knowledge, there is no dynamic version of these data structures that would further support deletions in time $o(\log n)$, let alone in constant amortized time.
  However it will be crucial in our algorithm to remove intervals.
  We thus accept to pay an extra logarithmic factor, and resolve to the simpler use of self-balancing binary search trees such as red-black trees~\cite{Cormen}.
  Red-black trees take $O(n \log n)$ to build (by $n$ successive insertions in time $O(\log n)$), and support search queries in $O(\log n + q)$ and deletions in $O(\log n)$.
  Here the search queries are of the form: ``list all nodes (intervals) containing a query element or intersecting a query interval''.

  We maintain two red-black trees.
  The first, $\mathcal T_{\mathcal B}$, is initialized to the $2 |\mathcal B|$ nodes of $\{A, B$  $|$ $(A,B) \in \mathcal B\}$, that is the sides of the bicliques of the \sbip.
  These intervals are sorted by lexicographic order on their pairs of endpoints.
  This tree will maintain which bicliques are still untraversed in a given direction (we will distinguish the two orientations).
  The second, $\mathcal T_U$, initially comprises the $n$ vertices in $V(G) = [n]$, sorted in the usual order.
  (The integers can be seen as singleton intervals to unify $\mathcal T_{\mathcal B}$ and $\mathcal T_U$ into the same kind of objects.)
  It will maintain which vertex of $G$ are still unexplored.

  The primitive $\texttt{Bel}(u,\mathcal T_{\mathcal B})$ (as in \emph{bel}ongs) reports all the biclique sides $S \in \mathcal T_{\mathcal B}$ such that $u \in S$, while $\texttt{Adj}(u,\mathcal T_{\mathcal B})$ (as in \emph{adj}acency) reports the set of biclique sides $B \in \mathcal T_{\mathcal B}$ such that there is a biclique $(A,B) \in \mathcal B$ with $u \in A$.
  Finally $\texttt{Int}(\mathcal T_U,[i,j])$ (as in \emph{int}ersection) lists all the elements of $\mathcal T_U$ that are in $[i,j]$, and we denote by $\text{delete}(u,\mathcal T)$ the deletion of $u$ from the red-black tree $\mathcal T$.

  We can now write our algorithm \sssp from a classic BFS, by replacing the access to edges of the current vertex $u$ by $\texttt{Adj}(u,\mathcal T_{\mathcal B})$, and the vertices to enqueue (and explore later) by $\bigcup_{[i,j] \in \texttt{Adj}(u,\mathcal T_{\mathcal B})} \texttt{Int}(U,[i,j])$.
  More precisely, we initialize a queue $Q$ to $\{s\}$, a set of unexplored vertices $U$ to $V(G) \setminus \{s\}$ as a red-black tree $\mathcal T_u$, a set of unaccessed biclique sides of $\mathcal B$ as another red-black tree $\mathcal T_{\mathcal B}$, a shortest-path tree parent relation $p$ by $p(s) := s$, and a distance table $d$ to the source $s$ by $d(s) := 0$.
  We remove $s$ from $\mathcal T_U$.
  As long as $Q$ is non-empty, we dequeue $u$ from it, and set $\mathcal S_u := \texttt{Bel}(u,\mathcal T_{\mathcal B})$, and $\mathcal N_u := \texttt{Adj}(u,\mathcal T_{\mathcal B})$.
  We remove all the biclique sides of $\mathcal S_u$ from $\mathcal T_{\mathcal B}$.
  We set $N_u := \bigcup_{[i,j] \in \mathcal N_u} \texttt{Int}(\mathcal T_U,[i,j])$.
  For every $v \in N_u$, we set $p(v)$ to $u$, $d(v)$ to $d(u)+1$, we enqueue $v$ in $Q$, and remove it from $\mathcal T_U$.
  We finally return $p$ and $d$ (see Algorithm~\ref{alg:sssp} for the pseudo-code).

  \begin{algorithm}
  \DontPrintSemicolon
  \SetKwInOut{Input}{Input}
  \SetKwInOut{Output}{Output}
  \Input{~~A graph $G$, a source $s \in V(G)$, and an \bip $\mathcal B$ of $G$.}
  \Output{~~A shortest-path tree $p$ rooted at $s$, with a distance table $d$ to $s$.}
  $\mathcal T_U \leftarrow V(G)$\;
  $\mathcal T_{\mathcal B} \leftarrow \{A, B$ $|$ $(A,B) \in \mathcal B\}$\;
  $Q \leftarrow \{s\}$\;
  $p(s) \leftarrow s$\;
  $d(s) \leftarrow 0$\;
  $\text{delete}(s,\mathcal T_U)$\;
 \While{$Q \neq \emptyset$}{
   $u \leftarrow \text{dequeue}(Q)$\;
   $\mathcal S_u \leftarrow \texttt{Bel}(u,\mathcal T_{\mathcal B})$\;
   $\mathcal N_u \leftarrow \texttt{Adj}(u,\mathcal T_{\mathcal B})$\;
   \For{$S \in \mathcal S_u$}{
     $\text{delete}(S,\mathcal T_{\mathcal B})$\;
   }
   $N_u \leftarrow \bigcup_{[i,j] \in \mathcal N_u} \texttt{Int}(\mathcal T_U,[i,j])$\;
   \For{$v \in N_u$}{
     $p(v) \leftarrow u$\;
     $d(v) \leftarrow d(u)+1$\;
     $\text{enqueue}(Q,v)$\;
     $\text{delete}(v,\mathcal T_U)$\;
   }
 }
 \Return{$p, d$}\;
 \caption{\algsssp}
 \label{alg:sssp}
  \end{algorithm}

  \medskip

  \textbf{Correctness.}
  Our algorithm is a BFS in which some edges that are not traversed may still disappear in one direction (line 12).
  We only need to argue that these arcs cannot be part of a shortest-path tree rooted at $s$.
  Say the current vertex is $u$, and the set of unexplored vertices is $U$ (i.e., the nodes of $\mathcal T_U$).
  We consider the set $\mathcal S_u := \texttt{Bel}(u,T_{\mathcal B})$ of biclique sides still in $\mathcal T_{\mathcal B}$ and containing $u$.
  All these intervals are then removed from $\mathcal T_{\mathcal B}$.
  Let $u' \neq u$ be a vertex in a side $S \in \mathcal S_u$, and let $S'$ be another side such that $(S,S') \in \mathcal B$.
  The deletion of $\mathcal S_u$ implies that an arc from $S$ to $S'$ can no longer be taken.
  We claim that it is safe to remove the arcs from $u'$ (or more generally from $S$) to $S'$.
  Indeed if $u'$ is visited after $u$, then $d(u) \leqslant d(u')$.
  Thus all the vertices in $N_u \supseteq S' \cap U$ have already had their distance set to $d(u)+1$ ($\leqslant d(u')+1$) and their parent set to $u$. 

  Note however that the biclique $(S',S)$ may still be traversed (in the other direction, from $S'$ to $S$).
  These arcs can very well be on a shortest-path tree.
  That is why we are removing biclique sides and not bicliques.
  
  \medskip

  \textbf{Running time.}
  The initialization of $\mathcal T_U$ and $\mathcal T_{\mathcal B}$ takes time $O(n \log n)$ and $O(|\mathcal B| \log |\mathcal B|) = O(|\mathcal B| \log n)$, respectively (observe that $|B| \leqslant n^2$, thus $O(\log |\mathcal B|) = O(\log n)$).
  Each call $\texttt{Bel}(u,T_{\mathcal B})$ reporting $q$ sides takes time $O(\log n + q)$.
  It is immediately followed by the deletion of these sides from $\mathcal T_{\mathcal B}$, in time $O(q \log n)$.
  Therefore in the entire while loop, these operations take overall time $O(|\mathcal B| \log n)$.
  Observe that $\texttt{Adj}(u,T_{\mathcal B})$ is built from $\texttt{Bel}(u,T_{\mathcal B})$ by simple access to the look-up table $\mathcal Z$ encoding $\mathcal B$.
  This takes time $O(|\texttt{Adj}(u,T_{\mathcal B})|)$.
  Since every biclique can be traversed at most twice (once in each direction), overall the calls $\texttt{Adj}(u,T_{\mathcal B})$ take time $O(|\mathcal B|)$.
  Each call $\texttt{Int}(\mathcal T_U,[i,j])$ reporting $q$ vertices takes time $O(\log n + q)$.
  This is followed by removing these vertices from $\mathcal T_U$ in time $O(q \log n)$.
  Hence this part takes overall time $O(n \log n)$.
  The rest of the instructions take constant time.
  Therefore the running time of \algsssp is $O((n+|\mathcal B|)\log n)$.
\end{proof}

As a direct corollary of \cref{lem:seq-to-bir,thm:bir-to-sp}, we get the following two theorems.
\begin{theorem}\label{thm:sssp}
  Let $\mathcal C$ be a class of bounded twin-width on which there is an $O_d(n \log n)$-time algorithm computing $d$-sequences for $n$-vertex graphs.
  Then \sssp can be solved in $\mathcal C$ in time $O_d(n \log n)$.
\end{theorem}

\begin{theorem}\label{thm:apsp}
  Let $\mathcal C$ be a class of bounded twin-width on which there is an $O_d(n^2 \log n)$-time algorithm computing $d$-sequences for $n$-vertex graphs.
  Then \apsp can be solved in $\mathcal C$ in time~$O_d(n^2 \log n)$.
\end{theorem}

Note that for all the classes shown to have bounded twin-width in the first two papers of the series~\cite{twin-width1,twin-width2}, an $O_d(n^2)$-time algorithm computes a $d$-sequence (where $d$ does not depend on $n$).
For some sparse classes ($K_t$-minor free graphs), or some dense classes sparsely presented (unit interval graphs, posets of bounded antichain), it is even possible to obtain the contraction sequence in time $O_d(n \log n)$.
For the latter kind, it yields $O(n \log n)$-time algorithms (that is, sublinear in the number of edges) computing shortest-path trees from a given source.
However in these individual classes, much simpler arguments would give $O(n)$-time algorithms.
Thus the strength of \cref{thm:sssp,thm:apsp} lies more in unifying and generalizing graph classes where $\Tilde{O}(n)$ and $\Tilde{O}(n^2)$ are achievable for \ssssp and \sapsp, and in the simplicity of the algorithm (a slightly modified BFS).

One could wonder if the diameter of a graph given with an $O(1)$-sequence can be computed significantly faster than in $O(n^2 \log n)$, by simply calling \sapsp and reporting the longest distance.
We observe that no truly subquadratic algorithm is possible, unless the Strong Exponential Time Hypothesis\footnote{The assumption that, for every $\varepsilon > 0$, \textsc{SAT} cannot be solved in time $(2-\varepsilon)^n$ by a classical algorithm.} (SETH) fails. 
\begin{theorem}\label{thm:diameter}
For every $\varepsilon, \varepsilon' > 0$, \textsc{Diameter} on bounded twin-width graphs cannot be computed, or $3/2 - \varepsilon'$-approximated, in time $n^{2-\varepsilon}$, unless the SETH fails, even if an $O(1)$-sequence of the input graph is given. 
\end{theorem}
\begin{proof}
  Such an SETH lower bound exists on graphs of bounded degree~(see~\cite{Evald16}).
  We subdivide $\ell - 1$ times each edge of a hard instance $H$, with degree bounded by $\Delta$ and $n' > 1$ vertices, where $\ell := \lceil \log n' \rceil$.
  We attach a pending path on $\ell$ edges to the $n'$ original vertices of $H$.
  This defines a graph $G$ with $n \leqslant \Delta/2 \cdot (\ell-1) n' + \ell n' = O(n' \log n')$ vertices.
  Thus $n = O(n'^{1+\frac{\varepsilon}{2}})$.
  We observe that $\diam(G) = \ell + \ell~\cdot~\diam(H) + \ell = (\ell+2) \diam(H)$.
  Besides we show in~\cite{twin-width2} that the $\log n'$-subdivision of $n'$-vertex graphs have bounded twin-width.
  Furthermore an $O(1)$-sequence can be computed in $O(n)$-time if the initial graph has bounded degree.
  An $n^{2-\varepsilon}$-time algorithm computing the diameter of such a graph $G$, would give an $O((n'^{1+\frac{\varepsilon}{2}})^{2-\varepsilon}) = O(n'^{2-\frac{\varepsilon^2}{2}})$.
  Such a subquadratic algorithm is ruled out, even to obtain a $3/2-\varepsilon'$-approximation of the diameter, unless the SETH fails.
  Finally one may observe that the reduction preserves the inapproximability gap.
\end{proof}

A related SETH lower bound is obtained by Coudert et al.~\cite{Coudert19}, who show that \textsc{Diameter} cannot be solved in time $2^{o(\text{cw})}n^{2-\varepsilon}$ on $n$-vertex graphs with clique-width $\text{cw}$.
The lower bound of \cref{thm:diameter} is quantitatively stronger (albeit in an admittedly larger graph class) since it rules out any algorithm solving \textsc{Diameter} in time $f(d)n^{2-\varepsilon}$ for any function $f$, on graphs of twin-width at most~$d$.
Let us recall that when the diameter is guaranteed constant, \textsc{Diameter} can be expressed as a first-order formula.
Thus we can compute the exact diameter in $O(n)$-time provided the contraction sequence of the input graph~\cite{twin-width1}.

\section{Approximation Algorithms}\label{sec:approx-alg}

Provided $O(1)$-sequences of the inputs, we give constant-approximation algorithms for \ds and the \dmis{2} problem, where one seeks a maximum-cardinality subset of vertices not containing a pair at distance at most 2.
Next we show that such an algorithm for \dmis{1}, that is \mis, would have the unexpected consequence of leading to a polynomial-time approximation scheme.

\subsection{Constant approximation for \ds}\label{subsec:approx-ds}

In this section, we prove that \ds and its dual \dmis{2} have bounded integrality gaps in classes of bounded twin-width.
Constant factor approximation algorithms follow for these two problems.
We will use the following technical lemma from the second paper of the series.

\begin{theorem}[Section 3, Lemma 20 in \cite{twin-width2}]\label{versatile}
  For every integer $t$, there are integers $s$ and $t'$ such that every graph~$G$ with a $t$-sequence admits a rooted tree $\mathcal T$ with the following properties.
  \begin{itemize}
  \item Every node of $\mathcal T$ is labeled by a $t'$-trigraph.
  \item The root of $\mathcal T$ is labeled by $G$.
  \item All the leaves of $\mathcal T$ are labeled by the 1-vertex graph $K_1$.
  \item If a node $x$ of $\mathcal T$ is labeled by $H$, and a child node of $x$ is labeled by $H'$, there is a $t'$-contraction in $H$ that yields $H'$.
    In particular $|V(H)|=|V(H')|+1$.
  \item Every internal node of $\mathcal T$ labeled by $H$ has at least $|V(H)|/s$ children coming from $t'$-contractions on pairwise disjoint pairs of vertices of $H$.
  \end{itemize}
\end{theorem}

Such a tree is called an \emph{$s$-versatile tree of $t'$-contractions}.
Informally \cref{versatile} says that, by degrading the twin-width bound, one can move away from the ``linear nature'' of the contraction sequence to a profusely branching contraction witness.

\cref{versatile} is effective: The $s$-versatile tree of $t'$-contractions can be computed in polynomial time, if a $t$-sequence for $G$ is provided.

\begin{theorem}\label{thm:dominating-gap}
  In classes of bounded twin-width, \ds has bounded integrality gap.
\end{theorem}
\begin{proof}
  Let $G$ be a graph of twin-width at most $t$.
  By~\cref{versatile}, there exist $t',s$ function of $t$ only such that $G$ admits an $s$-versatile tree of $t'$-contraction.
  Let $w^* : V(G) \to [0,1]$ be the weight function of a minimum fractional dominating set, with total weight $\gamma^*$.
  Thus $w^*$ is an optimum solution of the linear program
  \begin{align*}
  & \text{minimize } \sum_{x \in V(G)} w(x) \\
  & \text{with } \forall x \in V(G),\ \sum_{y \in N[x]} w(y) \ge 1, \text{ and } 0 \le w(x) \le 1,
  \end{align*}
  and $\gamma^* = \sum_{x \in V(G)} w^*(x)$.
  The weight function $w^*$ is extended to subsets of vertices by sum.
  We assume that $G$ has at least one vertex, so $\gamma^* \ge 1$.
  
  We now greedily perform contractions in $G$ following the versatile tree of contractions with a restriction:
  contractions involving a part of total weight at least $\frac{1}{2(t'+1)}$ are forbidden.
  Let us explain what this means in more detail.  
  We start at the root, labeled $G$, of the versatile tree.
  We move to a(ny) child node along an edge corresponding to a non-forbidden $t'$-contraction.
  A \emph{non-forbidden} contraction is one of $u, v$ with $w^*(u(G)) < \frac{1}{2(t'+1)}$ and $w^*(v(G)) < \frac{1}{2(t'+1)}$. 
  We iterate that until we get stuck (every child of the current node entails a forbidden contraction). 

  We adopt the partition viewpoint of the $t'$-sequence.
  Let $\mathcal P$ be the partition of $V(G)$ obtained when this process finishes, and let $G_{\mathcal P}$ be the corresponding trigraph (that is, the label of the node where we stop).
  We observe that we cannot end at a leaf of the versatile tree.
  Indeed that would mean that the last contraction merged a bipartition $\{X,Y\}$ of $V(G)$ into $\{V(G)\}$.
  As $\gamma^* \ge 1$, this would imply that $w^*(X) \geqslant 1/2$ or $w^*(Y) \geqslant 1/2$, contradicting $\max(w^*(X),w^*(Y)) < \frac{1}{2(t'+1)}$. 
  
  \begin{claim}\label{clm:partition-coarse}
    The partition $\mathcal P$ has at most $2s(t'+1)\gamma^*$ classes.
  \end{claim}
  \begin{proof}
    As we explained, we cannot end up with a partition $\mathcal P$ at a leaf of the versatile tree.
    Thus at least $\card{\mathcal P}/s$ disjoint pairs of vertices are $t'$-contractions in $G_{\mathcal P}$.
    Therefore all these contractions must be forbidden by our restriction imposed on the weights.
    It follows that at least $\card{\mathcal P}/s$ parts of $\mathcal P$ have weight at least $\frac{1}{2(t'+1)}$.
    Since the sum of all weights in $\mathcal P$ is $\gamma^*$, it follows that $\card{\mathcal P} \le 2s(t'+1)\gamma^*$.
  \end{proof}

  \begin{claim}\label{clm:singleton-weight}
    Let $P \in \mathcal P$ be any part.
    Either $w^*(P) < \frac{1}{t'+1}$ or $P$ is a singleton.
  \end{claim}
  \begin{proof}
    Let $P \in \mathcal P$, and assume that $P$ is not a singleton.
    Then $P$ has been obtained by contracting two parts $P_1, P_2$ during the contraction sequence leading to $\mathcal P$.
    The restriction on the contraction sequence ensures that $w^*(P_1) < \frac{1}{2(t'+1)}$ and $w^*(P_2) < \frac{1}{2(t'+1)}$.
    Therefore $w^*(P) = w^*(P_1) + w^*(P_2) < \frac{1}{t'+1}$.
  \end{proof}

  Let $D \subseteq V(G)$ be obtained by picking arbitrarily one vertex $x_P$ in each part $P \in \mathcal P$.
  By \cref{clm:partition-coarse}, $\card{D} \le 2s(t'+1)\gamma^*$, which is linear in $\gamma^*$ when $t$ is fixed.
  Let us prove that $D$ is a dominating set.
  We let $P \in \mathcal P$, and prove that all vertices of $P$ are dominated by $D$.

  Suppose first that there exists $P' \in \mathcal P$ such that $P, P'$ is a black edge in $G_{\mathcal P}$.
  Then $x_{P'} \in P'$ is adjacent to all vertices of $P$, which are thus dominated by $D$.

  Hence we may instead assume that $P$ does not have any black neighbor in $G_{\mathcal P}$.
  Consider any vertex $y \in P$, and let $P_1, \ldots, P_k$ the parts of $\mathcal P \setminus \{P\}$ such that there exists an edge between $y$ and some vertex of $P_i$.
  Then $P_1, \ldots, P_k$ are neighbors of $P$ in $G_{\mathcal P}$, and must be red neighbors since $P$ has no black neighbor.
  Since $G_{\mathcal P}$ is a $t'$-trigraph, it follows that $k \le t'$.
  
  We now claim that one of the parts $P, P_1, \ldots, P_k$ must be a singleton.
  Indeed, since $w^*$ is a fractional dominating set, and since $P \cup \bigcup_{i=1}^k P_i$ contains $y$ and its neighborhood, it must be that $w^*(P) + \sum_{i=1}^k w^*(P_i) \ge 1$.
  Because $k \le t'$, it follows that one part among $P, P_1, \ldots, P_k$ has weight at least $\frac{1}{t'+1}$.
  By~\cref{clm:singleton-weight}, that same part $P_h$ must be a singleton.
  Let $z$ be the single vertex in $P_h$.
  Necessarily $z \in D$.
  If this singleton part is $P$, then $z = y$.
  Otherwise $z$ is a neighbor of $y$ by definition of $P_1, \ldots, P_k$.
  In either case $y$ is dominated in $D$ by $z$.
\end{proof}

We now consider the following linear programming formulation of \dmis{2}, which is dual to \ds:
\begin{align*}
  & \text{maximize } \sum_{x \in V(G)} w(x) \\
  & \text{with } \forall x \in V(G),\ \sum_{y \in N[x]} w(y) \le 1, \text{ and } 0 \le w(x) \le 1.
\end{align*}
Similar arguments prove the same result for this dual problem.

\begin{theorem}\label{thm:2-independent-gap}
  In classes of bounded twin-width, \dmis{2} has bounded integrality gap.
\end{theorem}
\begin{proof}
  Consider $G$ of twin-width $t$, and $t',s$ function of $t$ such that $G$ admits an $s$-versatile tree of $t'$-contraction.
  Let $w^* : V(G) \to \mathbb{R}$ be the weight function of a maximum fractional 2-independent set, with total weight $\alpha_2^*$.

  We greedily perform contractions in $G$ following the versatile tree of contractions with the restriction:
  contractions involving a part with total weight more than 1 are forbidden.
  Let $\mathcal P$ be the partition of $V(G)$ obtained when this process finishes, and $G_{\mathcal P}$ be the corresponding trigraph.
  Again the weight function $w^*$ is extended to $\mathcal P$ by sum.
  With our restriction on allowed contractions, it is immediate that all classes of $\mathcal P$ have weight at most 2.
  Therefore $\card{\mathcal P} \ge \frac{\alpha_2^*}{2}$.
  We can safely assume that $\alpha_2^* > 2$, thus $\card{\mathcal P} > 1$.
  In particular, the node of the versatile tree labeled $G_{\mathcal P}$ in which we stopped is an internal node.

  Let $A = \{P \in \mathcal{P} \, : \, w^*(P) > 1\}$.
  \begin{claim}\label{clm:size-A}
    $\card{A} \ge \frac{\alpha_2^*}{2s}$.
  \end{claim}
  \begin{proof}
    The elements of $A$ are exactly the ones which cannot be used for contractions in $G_{\mathcal P}$.
    The versatile tree of contractions ensures at least $\card{\mathcal P}/s \ge \frac{\alpha_2^*}{2s}$ pairwise disjoint $t'$-contractions in $G_{\mathcal P}$.
    All these contractions must be forbidden, meaning that they all involve a vertex of~$A$.
    Since they are contractions of disjoint pairs of vertices, it follows that $\card{A} \ge \frac{\alpha_2^*}{2s}$.
  \end{proof}

  \begin{claim}\label{clm:no-black-edge}
    No element of $A$ has a black neighbor in $G_{\mathcal P}$.
  \end{claim}
  \begin{proof}
    Suppose that there exist $P \in A$, $P' \in \mathcal{P}$ such that $PP'$ is a black edge in $G_{\mathcal P}$.
    Then for any $x \in P'$ we have $P \subseteq N_G(x)$ and $w^*(P) > 1$, which violates the LP constraint.
  \end{proof}

  \begin{claim}\label{clm:A-2independent}
    There exists $S \subseteq A$ a 2-independent set in $G_{\mathcal P}$ such that $\card{S} \ge \frac{\alpha_2^*}{2s(t'^2+1)}$.
  \end{claim}
  \begin{proof}
    By \cref{clm:no-black-edge}, a path of length at most 2 in $G_{\mathcal{P}}$ between elements of $A$ can only consist of red edges.
    Since the red graph in $G_{\mathcal P}$ has maximum degree at most $t'$, given $P \in A$, there are at most $t'^2$ other elements of $A$ at distance 2 or less of $P$.
    Thus one can choose a 2-independent set in $A$ of size at least $\frac{\card{A}}{t'^2+1}$, which is at least $\frac{\alpha_2^*}{2s(t'^2+1)}$ by \cref{clm:size-A}.
  \end{proof}

  To conclude, we pick one vertex of $G$ within each part of $S$.
  This gives a 2-independent set in $G$ of size at least $\frac{\alpha_2^*}{2s(t'^2+1)}$.
\end{proof}

Reporting approximated solutions for \ds and \dmis{2} requires that a $t$-sequence of the input is provided (or that it can be computed in polynomial time, as it is the case on many bounded twin-width classes).
Interestingly, deciding the associated constant-gap problem can be done without $t$-sequences, with the mere knowledge of the twin-width bound. 

The constant approximations more generally work for \textsc{Min $r$-Dominating Set} and \dmis{$2r$}, for every positive integer $r$.
Indeed solving these problems in $G$ is equivalent to solving \ds and \dmis{2} in $G^{\le r}$ (where $G^{\le r}$ is the graph obtained by putting an edge between every pair of vertices at distance at most $r$ in $G$).
Besides the twin-width of $G^{\le r}$ is bounded by a function of the twin-width of $G$ and $r$, and an $O_r(1)$-sequence for $G^{\le r}$ can be computed in polynomial time, given an $O(1)$-sequence for~$G$~\mbox{\cite[Section 8, Theorem 41]{twin-width1}}.

\subsection{A constant approximation for \mis would imply a PTAS}\label{subsec:mis-barrier}

A pessimistic stance on the result of this section is that, perhaps surprisingly, the constant approximations of \ds and \dmis{2} are unlikely to be generalizable to the closely related \mis (that can be seen as \dmis{1}). 
We indeed observe that the self-improving reduction of Feige et al.~\cite{Feige91} preserves the twin-width.
As a consequence a constant approximation for \mis would provide a polynomial-time approximation scheme (PTAS).

\begin{theorem}
  \label{thm:MIS-APX-PTAS}
  If \lmis on graphs of twin-width at most $d$ has a constant-approximation algorithm, then it admits a PTAS.
\end{theorem}

For $G_1$ and $G_2$ two non-empty graphs, and $u \in V(G_1)$, we denote by $G_1(u \leftarrow G_2)$ the substitution in $G_1$ of $u$ by $G_2$.
That is, $u$ is replaced by $G_2$, and every vertex of $V(G_1) \setminus \{u\}$ initially adjacent to $u$ is made adjacent to the whole $V(G_2)$.

\begin{lemma}\label{lem:substitution}
$\tww(G_1(u \leftarrow G_2))=\max(\tww(G_1),\tww(G_2))$.
\end{lemma}
\begin{proof}
  We set $G := G_1(u \leftarrow G_2)$.
  $G_1$ and $G_2$ are both induced subgraphs of $G$, so $\tww(G) \geqslant \max(\tww(G_1),\tww(G_2))$.
  For the reverse inequality, one just applies the sequence of $d_2$-contractions on the copy of $G_2$ in $G$, with $d_2 := \tww(G_2)$.
  This results in the graph $G_1$ without red edges.
  Then, one applies the sequence of $d_1$-contractions to $G_1$, with $d_1 := \tww(G_1)$.
  This shows that $\tww(G) \leqslant \max(d_1,d_2)$.
\end{proof}

  For $G$ a graph, let $G^t$ be the graph on the vertex set $V(G)^t$, such that for $\bar{x} = (x_1,\ldots,x_t)$, $\bar{y} = (y_1,\ldots,y_t)$ distinct vertices, $\bar{x}\bar{y} \in E(G^t)$ if and only if $x_iy_i \in E(G)$ where $i$ is the smallest index such that $x_i \neq y_i$.
This definition can be restated inductively: $G^0$ is the 1-vertex graph, and $G^t$ is obtained from $G$ by substituting each vertex by a copy of $G^{t-1}$.
With the notations of the initial definition, for $x \in V(G)$, the set of vertices of $G^t$ of the form $(x,x_2,\ldots,x_t)$ is a copy isomorphic to $G^{t-1}$.

The following holds as a direct consequence of~\cref{lem:substitution}.
\begin{lemma}
  \label{lem:recursive-substitution}
  For any graph $G$ and integer $t>0$, $\tww(G^t) =~\tww(G)$.
\end{lemma}

We now show that the independence number of $G^t$ is tightly related to the one of $G$. 
\begin{lemma}
  \label{lem:recursive-substitution-independent}
  For any graph $G$, both following conditions hold.
  \begin{enumerate}
  \item Given any independent set of size $k$ in $G$, one can compute an independent of size~$k^t$ in~$G^t$, in time~$O(k^t)$.
  \item Given any independent set of size $k'$ in $G^t$, one can compute an independent of size at~least~$\sqrt[t]{k'}$ in~$G$, in time~$O(k')$.
  \end{enumerate}
\end{lemma}
\begin{proof}
  Let $I$ be an independent set in $G$.
  Then $I^t$ seen as a subset of $V(G)^t$ is an independent of $G^t$, which proves the first item.

  For the second item, let $I$ be an independent set in $G^t$ of size at least $r^t$.
  We define
  \[ I' := \{x \in V(G)~:~\exists x_2,\ldots,x_t,~(x,x_2,\ldots,x_t) \in I\}. \]
  Then $I'$ is an independent set in $G$.
  If $|I'| \geqslant r$, we are done.
  Otherwise, for each $x \in I'$, let
  \[ I_{x} := \{(x_2,\ldots,x_t) \in V(G)^{t-1}~:~(x,x_2,\ldots,x_t) \in I\}. \]
  For any $x$, $I_x$ is an independent set in $G^{t-1}$.
  Furthermore we have $\sum_{x \in I'} |I_x| = |I|$, $|I|=r^t$, and $|I'|<r$, hence there exists some $x \in I'$ such that $|I_x| \geqslant r^{t-1}$.
  By induction on $t$ we obtain an independent of size at least $r$ in $G$.
\end{proof}

As an immediate corollary, $\alpha(G^t) = \alpha(G)^t$ where, we recall, $\alpha(H)$ denotes the size of a maximum independent set in $H$.

\begin{proof}[Proof of~\cref{thm:MIS-APX-PTAS}]
  Assume there is a polynomial-time $\beta$-approximation for \mis on graphs of twin-width at most $d$.
  Let $G$ be a graph with twin-width at most $d$.
  By~\cref{lem:recursive-substitution} the algorithm can be ran on $G^t$ to obtain an independent set of size at least $\frac{\alpha(G^t)}{\beta} = \frac{\alpha(G)^t}{\beta}$.
  By~\cref{lem:recursive-substitution-independent}, this independent set in $G^t$ can be turned into an independent set in $G$ of size at least ${\alpha(G)}/{\sqrt[t]{\beta}}$.
  This gives a polynomial-time $\sqrt[t]{\beta}$-approximation for arbitrary $t$.
  Thus the approximation ratio can be made arbitrarily close to 1.
\end{proof}

\subsection{Linear Erd\H{o}s-P\'osa property}\label{subsec:approx-ds}

Given a $0,1$-matrix $M$, two natural integer programs naturally arise: One can ask for a minimum-weight $0,1$-vector $X_h$ such that $M \cdot X_h \geq 1$ or for a maximum-weight $0,1$-vector $Y_p$ such that $M^t \cdot Y_p\leq 1$.
In the usual representation of $M$ as a hypergraph $H$ where columns are vertices and rows are hyperedges (each row seen as an indicator vector of a subset of vertices), $X_h$ is a minimum \emph{hitting set} and $Y_p$ is a maximum \emph{packing}.
We usually denote by $\mu(H)$ the size of a maximum packing and by $\tau(H)$ the size of a minimum hitting set. 

One can then consider the fractional relaxation of these parameters, $\mu^*(H)$ and $\tau^*(H)$. 
Since the corresponding linear programs are dual, we obtain the following chain of (in)equalities $\mu(H) \leq \mu^*(H) = \tau^*(H) \leq \tau(H)$.
A class $\mathcal H$ of hypergraphs for which there exists a function $f$ such that every hypergraph $H \in \mathcal H$ satisfies $\tau(H) \leq f(\mu(H))$ has the \emph{Erd\H{o}s-P\'osa property}.
If furthermore $\tau(H) \leq c \cdot \mu(H)$ for some constant $c$, $\mathcal H$ has the \emph{linear Erd\H{o}s-P\'osa property}.

By a result of Haussler and Welzl~\cite{hw-esrq-87}, the class of hypergraphs with bounded VC-dimension satisfies that $\tau(H) \leq f(\tau^*(H))$, but is not by itself sufficient to imply the Erd\H{o}s-P\'osa property (the integrality gap for $\mu$ is unbounded).
A result of Ding et al.~\cite{DBLP:journals/combinatorica/DingSW94} asserts that the Erd\H{o}s-P\'osa property holds for matrices which do not contain the transpose of incidence matrices of cliques as submatrices; the function $f$ is polynomial but not linear.
Dvoř\'ak~\cite{DVORAK2013833} proved that, for every fixed $r$, $r$-neighborhood hypergraphs of bounded expansion classes have the linear Erd\H{o}s-P\'osa property.
Recently, Bousquet et al~\cite{bousquet2020packing} showed that ball hypergraphs (of any radius) of proper minor-closed classes have the linear Erd\H{o}s-P\'osa property.

The \emph{incidence bipartite graph} $B(H)$ of a hypergraph $H$ is the bipartite graph on vertex set $V(H) \cup E(H)$ where $ve$ is an edge if $v \in V(H)$, $e \in E(H)$ and $v \in e$.
The \emph{twin-width} of hypergraph $H$ is defined here as the twin-width of $B(H)$.
A straightforward adaptation of the proofs of~\cref{thm:dominating-gap,thm:2-independent-gap} gives:

\begin{theorem}\label{thm:linearEP}
For every integer $t$, there is a constant $c_t$ such that every hypergraph $H$ with twin-width at most $t$ satisfies $\tau(H) \leq c_t \cdot \mu(H)$.
\end{theorem} 

In other words, the class of bounded twin-width hypergraphs have the linear Erd\H{o}s-P\'osa property.
A particularly interesting line of research would be to generalize this integrality-gap result to integer matrices rather than just $0,1$-matrices.
This requires a suitable definition for bounded twin-width in the general integer case.

\section{Future work and open questions}\label{sec:conclusion}

We have now a rather fine-grained understanding of the classic parameterized graph problems (\kmis, \kds, and their relatives) when a contraction sequence is given in addition to the bounded twin-width graph.
For \kmis for example there is a $2^{O(k)}n$-time algorithm, while a $2^{o(k / \log k)}n^{O(1)}$-time (even $2^{o(n / \log n)}$-time) algorithm would refute the ETH.
It is natural to wonder if better approximation algorithms of NP-hard problems are possible when a contraction sequence is given. 
Before we detail that a bit, as well as the possibility of getting improved exact exponential algorithms on general graphs, we note that bounded twin-width does not seem to help to get polynomial kernels.

\subsection{No polynomial kernels on bounded twin-width classes}\label{subsec:nopolyker}

We already observed that \kmis is unlikely to have $k^{O(1)}$ kernels on graphs of twin-width at most a fixed constant $d$~\cite{twin-width1}.
We sketch here that the same applies to the vertex-weighted \kds (that is, the problem of the existence of a weight-$k$ dominating set).
The following is an OR-composition producing from, say, $t$ instances of the NP-hard \textsc{Dominating Set} on planar graphs, one instance of \textsc{Weighted Dominating Set} whose underlying unweighted graph has constant twin-width.

We make the disjoint union of the $t$ planar \textsc{Dominating Set}-instances $(G_1,k), \ldots, (G_t,k)$.
We add $t$ vertices $u_1, \ldots, u_t$ each of weight $k+1$, and link $u_i$ to all the vertices of every $G_j$ but $G_i$.
It is easy to see that the existence of a weight-$2k+1$ dominating set in this new graph is equivalent to one of the instances $(G_1,k) \ldots, (G_t,k)$ being positive.
As planar graphs have bounded twin-width~\cite{twin-width1}, the built graph (forgetting its weights) also has bounded twin-width.
One can first contract every $G_i$ into single vertices, thus obtaining the (black) anti-matching on $t$ edges (i.e., the bipartite complement of $t$ independent edges), which itself has twin-width~2.
Thus a polynomial kernel would imply the unlikely containment NP $\subseteq$ co-NP/poly~\cite{Bodlaender09}.
It is not so satisfactory that the lower bound is for \textsc{Weighted Dominating Set}, while the twin-width is computed on the unweighted graph.
It turns out that the same negative result is attainable for \textsc{Dominating Set} but the reduction is far more involved.
Thus we will not sketch it here.

\subsection{Better approximation algorithms}\label{subsec:conclusion-app}

We ask for the approximability status of \lmis, \ds, and \textsc{Min Coloring} on bounded twin-width graphs (given with $d$-sequences).

One can observe that the arguments of~\cref{subsec:mis-barrier} show that a $\log^c n$-approximation algorithm for \mis (for some constant $c$) implies a $\log^{\varepsilon} n$-approximation for any $\varepsilon > 0$.
We let the reader decide if this is a sign that $\log^c n$-approximation algorithms are unlikely.
Approximation algorithms of \mis on bounded twin-width graphs with worst ratios (for instance $n^\varepsilon$ for every $\varepsilon > 0$) would also be interesting, as they are far from existing in general graphs.
For \ds on bounded twin-width graphs, we ask for a constant-approximation algorithm with ratio independent on the twin-width bound, or even for a PTAS.
For~\textsc{Min Coloring}, we ask for any improvement over our $2^{O(\text{OPT})}$-approximation algorithm.
A first step is to reach approximation factor $\text{OPT}^{O(1)}$.
While we do not see any obvious obstruction to an $O_d(1)$-approximation, a PTAS is ruled out by the 3 vs 4 hardness of \textsc{Coloring} in planar graphs (class for which $d$-sequences can be computed in polynomial time~\cite{twin-width1}). 

\subsection{Exact exponential algorithms}\label{subsec:exact-exp}

A possible algorithmic success for a novel graph invariant, like twin-width, is to eventually lead to (faster) algorithms on general graphs, and not merely on graphs where the invariant is bounded.
A natural way this happens (for instance for treewidth) is by a win-win argument.
Either the parameter is small and we exploit it, or it is large, and some complex structure appears, which actually helps our decision.

But win-win arguments are not the only way.
Algorithms initially designed for bounded twin-width graphs may turn out also interesting on general graphs.
We see \cref{thm:ctk-ct} as a promising starting point to get exact exponential algorithms for \lmis on general graphs.
This asks for a new game related to, but also fundamentally different from twin-width.
Can we find a contraction sequence for any $n$-vertex graph such that the total number of connected sets in the red graphs is at most $O^*(c^n)$ for some constant $c$?
(Showing this result with $c=1.19$ would improve the current best exact algorithm for \mis.)
Note that creating vertices with large red degree is no longer forbidden.

\end{document}